\documentclass[12pt]{article}

\usepackage{latexsym,amssymb,amsmath,enumerate,float,geometry,cite,subcaption,amsthm}

\usepackage[colorlinks = true,citecolor = blue,linkcolor=magenta]{hyperref}

\usepackage{enumitem}
\setlist{nolistsep}

\usepackage{graphicx}
\usepackage{tikz,rotating}
\usetikzlibrary{positioning}
\tikzset{main node/.style={draw,fill=white,circle,very thin,auto=left,inner sep=1pt}}

\newtheorem{theorem}{Theorem}
\newtheorem{corollary}[theorem]{Corollary}

\newtheorem{lemma}[theorem]{Lemma}

\newtheorem{definition}{Definition}
\newtheorem{question}{Question}

\newcommand{\BM}{{\sc BM }}
\usepackage{setspace}
\usepackage{mathptmx}      

\newcommand{\probld}[3]{
	\begin{flushleft}
		\vspace{-.15cm}
		\fbox{
			\begin{minipage}{.98\textwidth}
				\noindent {\sc #1}\\
				{\bf Input:} #2\\
				{\bf Question:} #3
			\end{minipage}
		}
		\vspace{-.15cm}
	\end{flushleft}
}

\begin{document}
	
	\title{On the Computational Complexity of the\\ Bipartizing Matching Problem\footnotemark \footnotetext{This study was financed in part by the Coordena\c{c}\~ao de Aperfei\c{c}oamento de Pessoal de N\'ivel Superior - Brasil (CAPES) 
			Finance Code 001, by the Conselho Nacional de Desenvolvimento Cient\'ifico e Tecnol\'ogico - Brasil (CNPq) - CNPq/DAAD2015SWE/290021/2015-4, and FAPERJ.\\
			A conference version appeared in the Proc. of the 12th Annual International Conference on Combinatorial Optimization and Applications (COCOA), Volume 11346, pages 198--213, Atlanta, USA, December 2018.}
	}
	\author{Carlos V.G.C. Lima$^{1,\dagger}$, Dieter Rautenbach$^{2,\dagger}$, U\'{e}verton S. Souza$^{3,\dagger}$,\\ and Jayme L. Szwarcfiter$^{4}$\thanks{\textbf{Email addresses:} carloslima@dcc.ufmg.br (Lima), dieter.rautenbach@uni-ulm.de (Rautenbach), ueverton@ic.uff.br (Souza), and jayme@cos.ufrj.br (Szwarcfiter)}
	}
	\date{}\maketitle
	\begin{center}
		{\small
			$^1$ Departamento de Ciência da Computação, Universidade Federal de Minas Gerais, Belo Horizonte, Brazil.\\
			$^2$ Institute of Optimization and Operations Research, Ulm University, Ulm, Germany.\\
			$^3$ Instituto de Computação, Universidade Federal Fluminense, Niter\'{o}i, Brazil.\\
			$^4$ PESC, COPPE, Universidade Federal do Rio de Janeiro, Rio de Janeiro, Brazil.
		}
	\end{center}
	
	\begin{abstract}
		We study the problem of determining whether a given graph~$G=(V,E)$ admits a matching~$M$ whose	removal destroys all odd cycles of~$G$ (or equivalently whether~$G-M$ is bipartite).
		This problem is equivalent to determine whether~$G$ admits a~$(2,1)$-coloring, which is a~$2$-coloring of~$V(G)$ such that each color class induces a graph of maximum degree at most~$1$.
		We determine a dichotomy related to the~{\sf NP}-completeness of this problem, where we show that it is~{\sf NP}-complete even for $3$-colorable planar graphs of maximum degree~$4$, while it is known
		that the problem can be solved in polynomial time for graphs of maximum degree at most~$3$.
		In addition we present polynomial-time algorithms for some graph classes, including graphs in which every odd cycle is a triangle, graphs of small dominating sets, and~$P_5$-free graphs.
		Additionally, we show that the problem is fixed parameter tractable when parameterized by the clique-width, which implies polynomial-time solution for many interesting graph classes,
		such as distance-hereditary, outerplanar, and chordal graphs.
		Finally, an~$O\left(2^{O\left(vc(G)\right)} \cdot n\right)$-time algorithm and a kernel of at most~$2\cdot nd(G)$ vertices are presented, where~$vc(G)$ and~$nd(G)$ are the	vertex cover number
		and	the neighborhood diversity of~$G$, respectively.\\
		
	\end{abstract}

\section{Introduction}
\label{sec:intro}

	Given a graph~$G=(V, E)$ and a graph property~$\Pi$, the \textit{$\Pi$ edge-deletion problem} consists in determining the minimum number of edges required to be removed in order to
	obtain a graph satisfying~$\Pi$~\cite{bbd}.
	Given an integer~$k \geq 0$, the~\textit{$\Pi$ edge-deletion decision problem} asks for a set~$F\subseteq E(G)$
	with~$|F| \leq k$, such that the obtained graph by the removal of~$F$ satisfies~$\Pi$.
	Both versions have received widely attention on the study of their computational complexity, where we can cite~\cite{y2,y,bbd,nss,as,gjs,gfpp,Lima17}.
	They are particular cases of the general branch of graph editing problems.
	In such problems some editing operations on the vertices or edges are allowed in order to obtain a new graph satisfying the desired property.
	They include vertex and edge additions and deletions, as in the problem studied in this paper, flipping and contraction of edges.
	Examples of applications and variations of graph modification problems can be found, for example, in~\cite{Chuangpishit}.
	
	When the obtained graph is required to be bipartite, the corresponding edge (resp. vertex) deletion problem is called \textit{edge \textnormal{(resp.} vertex\textnormal{)} bipartization}~\cite{cnr,Abdullah,fkr} or
	\textit{edge \textnormal{(resp.} vertex\textnormal{)} frustration}~\cite{yr,garcia16}.
	
	Choi et al.~\cite{cnr} showed that edge bipartization is {\sf NP}-complete even for cubic graphs.
	Furma\'nczyk et al.~\cite{fkr} considered vertex bipartization of cubic graphs by removing an independent set.
	This problem has also been considered by Bonamy et al.~\cite{Bonamy2018}, where it is called {\sc Independent Odd Cycle Transversal}.
	A generalization of this problem is given by Agrawal et al.~\cite{Agrawal}, which generalizes cycle hitting problems.
	Given a graph~$G$, a subgraph~$H$ of~$G$, and a positive integer~$k$, the goal is to decide whether there exists a vertex subset~$S$ of~$G$ that intersects all desired cycles and~$S$ is an independent set in~$H$.
	They called this problem as \textsc{Conflict Free Feedback Vertex Set} and studied it from the view-point of parameterized complexity.
	
	In this paper we study the similar version for the edge deletion decision problem for odd cycles.
	That is, the problem of determining whether	a finite, simple, and undirected graph~$G$ admits a removal of an edge set that is a matching in~$G$ in order to obtain a bipartite graph.
	Formally, for a set~$S \subseteq E(G)$, let~$G-S$ be the graph with vertex set~$V(G)$ and edge set~$E(G)\setminus S$.
	We say that a matching~$M \subseteq E(G)$ is a {\it bipartizing matching} of~$G$ if~$G-M$ is bipartite.
	Denoting by~$\mathcal{BM}$ the family of all graphs admitting a bipartizing matching, we deal with the computational complexity of the following decision problem.
	
	\probld	
	{Bipartizing Matching (BM)}
	{A finite, simple, and undirected graph~$G$.}
	{Does~$G \in \mathcal{BM}$?}
	
	A more restricted version was considered by Schaefer~\cite{Schaefer}, which asked whether a given graph~$G$ admits a $2$-coloring of the vertices such that each vertex has \emph{exactly} one neighbor with same color as itself.
	We can see that the removal of the set of edges whose endvertices have same color, which is a perfect matching of~$G$, generates a bipartite graph.
	He proved that this problem is~{\sf NP}-complete even for planar cubic graphs.
	
	With respect to the minimization version, where the set of removed edges is not required to be a matching, the edge-deletion decision problem in order to obtain a bipartite graph is analogous to the \textsc{Simple Max Cut} problem, which was proved to be~{\sf NP}-complete by Garey et al.~\cite{gjs}.
	Yannakakis~\cite{y2} proved its~{\sf NP}-completeness even for cubic graphs.
	
	
	\BM can also be seen as a \textit{defective coloring}~\cite{Cowen97}.
	A~$(k, d)$-coloring of a graph~$G$ is a $k$-coloring of~$V(G)$ such that each vertex has at most~$d$ neighbors with same color as itself.
	Defective colorings were introduced independently by Andrews and Jacobson~\cite{Andrews85}, Harary and Jones~\cite{Harary85}, and Cowen et al.~\cite{Cowen86}, which received wide attention in the literature~\cite{Angelini17,Eaton99,bky,Cowen97,Axenovich}.
	We can see that any proper coloring is a $(k,0)$-coloring, for some~$k \geq 1$.
	Moreover, a graph belongs to~$\mathcal{BM}$ if and only if it admits a~$(2,1)$-coloring.
	
	Lov\'asz~\cite{lovasz66} proved that if a graph~$G$ satisfies~$(d_1+1) + (d_2+1) + \dots + (d_k+1) \geq \Delta(G)+1$, then~$V(G)$ can be partitioned into~$V_1, \dotsc, V_k$, such that each \textit{induced subgraph}~$G[V_i]$
	has maximum degree at most~$d_i$, $1\leq i \leq k$, where~$\Delta(G)$ is the \textit{maximum degree of~$G$}.
	A proof of this result can be found in~\cite{Cowen97}.
	This result implies that all \textit{subcubic graphs}~$G$, graphs satisfying~$\Delta(G) \leq 3$, are~$(2,1)$-colorable.
	
	Cowen et al.~\cite{Cowen97} proved that it is~{\sf NP}-complete to determine whether a given graph is~$(2,1)$-colorable, even for graphs of maximum degree~4 and even for planar graphs of maximum degree~5.
	Angelini et al.~\cite{Angelini17} presented a linear-time algorithm which determines that partial~$2$-trees, a subclass of planar graphs, are~$(2,1)$-colorable.
	We emphasize that a~$k$-tree has treewidth at most~$k$, for any~$k\geq 1$.
	
	Eaton and Hull~\cite{Eaton99} proved that all triangle-free outerplanar graphs are also~$(2,1)$-colorable.
	Borodin et al.~\cite{bky} studied graphs with respect to a sparseness parameter, the \textit{average degree}
	$$mad(G) = \max \left\{\frac{2|E(H)|}{|V(H)|}\textnormal{, for all } H \subseteq G\right\}.$$
	They proved that every graph~$G$ with~$mad(G) \leq \frac{14}{5}$ is~$(2,1)$-colorable, where this bound is sharp.
	Moreover, they defined the parameter~$\rho(G) = \min_{S \subseteq V(G)}\rho_{G}(S)$, such that~$\rho_G(S) = 7|S|-5|E(G[S])|$.
	They showed that~$G$ is~$(2,1)$-colorable if~$\rho(G) \geq 0$.
	By the Euler's formula, a planar graph~$G$ with \textit{girth}~$g$, the size of a smallest cycle, has~$mad(G) < \frac{2g}{g-2}$.
	Hence if~$G$ has girth at least~$7$, then it is~$(2,1)$-colorable.
	
	The similar problem where the opposite condition on the degrees is considered for each part of the bipartition, that is, each part induces a subgraph of minimum degree at least a given positive integer, is studied by Bang-Jensen and Bessy~\cite{Bang19}.
	
	Another variations are given in~\cite{Lima17,Fabio18}, where the resulting graph by the removal of a matching is required to be a forest, eliminating even the even cycles of the input graph.
	As these works, our study in this paper consists in trying to find a bipartizing matching instead of trying improperly color the graph.
		
	\subsection{Our Results}
	\label{subsec:Results}
	
	\subsubsection{{\sf NP}-completeness}
	
	In this paper we extend the results of Cowen et al.~\cite{Cowen97}.
	As we said before, they proved the {\sf NP}-completeness of determining whether a given graph is~$(2,1)$-colorable, even for graphs of maximum degree~4 and for planar graphs of maximum degree~5.
	However, the complexity of~$(2,1)$-color for planar graphs of maximum degree~4 has been left open, which is our first contribution.
	\begin{theorem}\label{thm:np-planar4}
		\BM is {\sf NP}-complete for $3$-colorable planar graphs of maximum degree~4.
	\end{theorem}
	
	More precisely, Lov\'asz~\cite{lovasz66} proved the following theorem that appears many times in the literature and whose proof can be found in~\cite{Cowen97}.
	\begin{theorem}[\cite{lovasz66}]\label{thm:Lovasz}
		For any integer~$k>0$, any graph of maximum degree~$\Delta$ admits a~$\left(k,\lfloor\Delta/k \rfloor\right)$-coloring.
	\end{theorem}

	While Theorem~\ref{thm:Lovasz} ensures that all graphs of maximum degree~$\Delta = 2d+1$ admit a~$(2,d)$-coloring, Cowen et al.~\cite{Cowen97} proved that deciding whether a given graph of maximum degree~$2(d+1)$ is~$(2, d)$-colorable is~{\sf NP}-complete.
	Therefore they proposed the following question for the general case.
	\begin{question}[\cite{Cowen97}]\label{quest:original}
		In general, what is the complexity of deciding whether a given graph of maximum degree~$k(d+1)$ admits a~$(k, d)$-coloring?
	\end{question}
	
	Angelini et al.~\cite{Angelini17} answered Question~\ref{quest:original} for~$(3,1)$-coloring, showing that it is {\sf NP}-complete for graphs of maximum degree~6.
	However, they were unable to answer Question~\ref{quest:original} restricted to planar graphs, that is, to give the complexity of deciding whether a planar graph of maximum degree~$6$ admits a~$(3,1)$-coloring, proving the {\sf NP}-completeness for planar graphs of maximum degree~7.
	
	Cowen et al.~\cite{Cowen97} also proved that, for any integer~$d>0$, deciding whether a planar graph of maximum degree~$4d+3$ is~$(2,d)$-colorable is {\sf NP}-complete.
	Theorem~\ref{thm:np-planar4} shows that this bound on the maximum degree is not tight.
	
	Dorbec et al.~\cite{Dorbec14} proved the {\sf NP}-completeness of~$(2,2)$-coloring for planar graphs of maximum degree~4 and with no~$C_4$, but restricted to the case of bipartising the graph by the removal of 3-vertex star forests.
	Such defective colorings, where the induced subgraphs by each part are required to be forests, are known as \textit{(star, $k$)-coloring}~\cite{Angelini17}, when the size of each star is not bounded, or~\textit{$d$-star $k$-partition}~\cite{Dorbec14}, in the case that each star has at most~$d$ vertices.
	We emphasize that such defective colorings are restricted cases of $(k,d)$-colorings.
	
	The family of \textit{cographs} contains the single vertex~$K_1$ and is closed with respect to disjoint union and complementation.
	The \textit{cograph $k$-partition} of a graph~$G$ is a $k$-coloring of~$V(G)$ into color classes that are cographs.
	Note that the cographs can have unbounded maximum degree in each color class and the~(star, $k$)-coloring and $d$-star $k$-partition are particular cases of cograph $k$-partition.
	Gimbel and Ne\v{s}et\v{r}il~\cite{Gimbel10} proved that deciding whether a planar graph of maximum degree at most~6 is cograph $2$-partitionable is {\sf NP}-complete.
	However, each component of the bipartition can have a cograph of maximum degree greater than~$2$.
	
	In the proof of the upper bound~$4d+3$ on the maximum degree, Cowen et al.~\cite{Cowen97} considered a $d$-star $2$-partition of the graph, which also indicates that this bound can be improved.
	Note that Theorem~\ref{thm:Lovasz} establishes that all graphs of maximum degree at most~5 are~$(2,2)$-colorable.
	In this paper we also prove the following theorem based on Theorem~\ref{thm:np-planar4}.
	\begin{theorem}\label{thm:2,2-coloring}
		Deciding whether a planar graph of maximum degree~$6$ admits a $(2,2)$-coloring, that is not necessarily a cograph 2-partition, is {\sf NP}-complete.
	\end{theorem}
	
	Theorem~\ref{thm:2,2-coloring} is less restrictive than the result of Gimbel and Ne\v{s}et\v{r}il~\cite{Gimbel10}, being the best possible with respect to the maximum degree by Theorem~\ref{thm:Lovasz}.
	In the same way, we use Theorem~\ref{thm:2,2-coloring} in order to prove a similar result for~$(2,3)$-color planar graphs.
	\begin{corollary}\label{cor:2,3-coloring}
		It is {\sf NP}-complete to decide whether a planar graph of maximum degree~$8$ is $(2,3)$-colorable.
	\end{corollary}
	
	Theorems~\ref{thm:np-planar4}, ~\ref{thm:Lovasz}, ~\ref{thm:2,2-coloring}, and Corollary~\ref{cor:2,3-coloring} show a dichotomy with respect to the computational complexity of $(2,d)$-color planar graphs of bounded maximum degree for~$d \leq 3$.
	More precisely, they suggest that the bound of~$4d+3$ on the maximum degree to~$(2,d)$-color planar graphs can be improved.
	Unfortunately, our construction on the proofs cannot be extended for~$d \geq 4$ in a tight way, but is still better than the upper bound~$4d+3$ for~$d \leq 7$.
	Therefore, in order to generalize Question~\ref{quest:original} for planar graphs, we propose an initial one restricted to~$(2,d)$-color planar graphs.
	\begin{question}\label{quest:2,d-coloring}
		What is the minimum integer~$\Delta^*$, $2d+2 \leq \Delta^* \leq 4d+3$, such that deciding whether a given planar graph of maximum degree~$\Delta^*$ admits a~$(2, d)$-coloring is {\sf NP}-complete?
	\end{question}
	
	\subsubsection{Positive Results}
	
	On the positive side, we present polynomial-time algorithms for  on several graph classes.
	\begin{theorem}\label{thm:polynomial}
		{\sc BM} can be solved in polynomial time for the following graph classes:
		\begin{itemize}
			\item[\textnormal{(a)}] graphs having bounded dominating set;
			\item[\textnormal{(b)}] $P_5$-free graphs.
			\item[\textnormal{(c)}] graphs in which every odd-cycle subgraph is a triangle;
		\end{itemize}				
	\end{theorem}	
	We also study parameterized complexity aspects.
	Using Courcelle's meta-theorems~\cite{C90,C93,C97} we prove that \BM is fixed-parameter tractable when parameterized by the clique-width, which improves the previous result of Angelini et al.~\cite{Angelini17}.
	\begin{theorem}\label{thm:cwd}
		{\sc BM} is {\sf FPT} when parameterized by the clique-width.
	\end{theorem}
	By Theorem~\ref{thm:cwd} we can solve {\sc BM} for several interesting graph classes in polynomial time, as for example \textit{distance-hereditary}, \textit{series-parallel}, \textit{control-flow}, and some subclasses of planar graphs such as \textit{outerplanar}, \textit{Halin}, and \textit{Apollonian networks}~\cite{Br05,G00,B98,T98}.
	The same follows for~($P_6$, claw)-free and (claw, co-claw)-free graphs~\cite{bkm06,bell06}.
	Moreover, since clique-width generalizes several graph parameters~\cite{L12}, it follows that {\sc BM} is in FPT when parameterized by the following parameters: neighborhood diversity; treewidth; pathwidth; feedback vertex set; and vertex cover.
	In addition, it also follows that
	\begin{corollary} \label{cor:chordal}
		{\sc BM} is polynomial-time solvable for chordal graphs.
	\end{corollary}
	
	We also show an exact~$2^{O(vc(G))}\cdot n$ algorithm, where~$vc(G)$ is the \textit{vertex cover number} of~$G$.
	Finally, for a generalization of \BM, we show a kernel with at most~$2 \cdot nd(G)$ vertices when a more general problem is parameterized by \textit{neighborhood diversity number}, $nd(G)$.			
	
	\subsubsection{Organization of the Paper}
	
	We summarize our results as follows.
	In Section~\ref{sec:preliminaries} we present some basic definitions and notation used throughout the paper and some initial properties of graphs in~$\mathcal{BM}$.
	In Section~\ref{sec:NPc} we present the proof of Theorem~\ref{thm:np-planar4}.
	In Section~\ref{sec:PolyResults} we present the proofs of the statements in Theorems~\ref{thm:polynomial}.
	In Section~\ref{sec:FixedTract} we consider the parameterized complexity aspects and present the proof of Theorem~\ref{thm:cwd}.
	In Section~\ref{sec:2,d-coloring} we consider the complexity of~$(2,d)$-color planar graphs of bounded degree, for~$d\geq 2$, and present the proofs of Theorem~\ref{thm:2,2-coloring} and Corollary~\ref{cor:2,3-coloring}.
	Finally in Section~\ref{sec:conclusion} we present the conclusions.
	
	\section{Preliminaries}
	\label{sec:preliminaries}
	
	\subsection{Basic Definitions and First Remarks}
	\label{subsec:defProp}
	
	Now we present some basic definitions, notations, and some initial remarks.
	We use standard notation and definitions of graph theory, where we consider only simple and undirected graphs.
	For more details on graph terminology and notation, see~\cite{Diestel10}.
	
	Given a graph~$G=(V,E)$, we denote by~$n(G)$ and~$m(G)$ the number of vertices and edges of~$G$, respectively.
	For a vertex~$v \in V(G)$, let~$N_G(v)$ be the \textit{neighborhood} of~$v$ in~$G$ and~$N_G[v] = \{v\} \cup N_G(v)$ its \textit{closed neighborhood} in~$G$.
	The \textit{degree} of~$v \in V(G)$,~$|N_G(v)|$, is denoted as~$d_G(v)$.
	The subscripts can change for a subgraph~$H$ of~$G$ when necessary.
	
	Given a set~$S \in V(G)$, let~$H = G[S]$ be the \textit{induced subgraph} of~$G$ by~$S$, such that~$V(H)=S$ and~$uv \in E(H)$ if and only if~$uv \in E(G)$ and~$u, v \in S$.
	We also say that~$S$ \textit{induces}~$H$ and that~$H$ is the graph induced by~$S$.
	
	Let~$P_n = v_1v_2 \dots v_n$ and~$C_n = v_1v_2 \dots v_nv_1$ be the induced \textit{path} and induced \textit{cycle} of order~$n$, respectively.
	Furthermore, we denote by~$K_n$ and~$K_{n, m}$ the \textit{complete graph} of order~$n$ and the \textit{complete bipartite graph} with parts of order~$n$ and~$m$, respectively.
	A \textit{diamond} is the graph obtained by removing one edge from the~$K_4$.
	
	A \textit{universal vertex}~$v$ of a graph~$G$ is one adjacent to all vertices in~$V(G)\setminus \{v\}$.
	For an integer~$k \geq 3$, let~$W_k$ denote the \textit{wheel graph} of order~$k+1$, that is, the graph obtained by connecting a universal vertex, called \textit{central}, to all the vertices of an induced cycle~$C_k$.
	The~$W_4$ and~$W_5$ are depicted in Figures~\ref{fig:W_4} and~\ref{fig:W_5}, respectively.
	
	We say that a graph~$G$ is a~\textit{$k$-pool} if it is composed by~$k$-edge disjoint triangles whose union of all the vertices	of their bases induces a~$C_k$.
	Formally, a~$k$-pool is obtained from an induced cycle~$C_{2k}=v_1 v_2 \dots v_{2k}v_1$, $k \geq 3$, such that the vertices	with index odd induce the \textit{internal cycle}~$p_1p_2 \dots p_{k}p_1$ of the~$k$-pool, where~$p_i = v_{2i-1}$, $1\leq i \leq k$.
	The vertex with index even~$b_i = v_{2i}$, $1\leq i \leq k$, is called the \textit{$i$-th-border} of the~$k$-pool.
	For~$G$ a~$k$-pool, if~$k$ is odd, then we say that~$G$ is an \textit{odd $k$-pool}, otherwise we say that it is an \textit{even $k$-pool}. 
	The~$3$-pool and~$5$-pool are depicted in Figures~\ref{fig:3-pool} and~\ref{fig:5-pool}, respectively.
	
	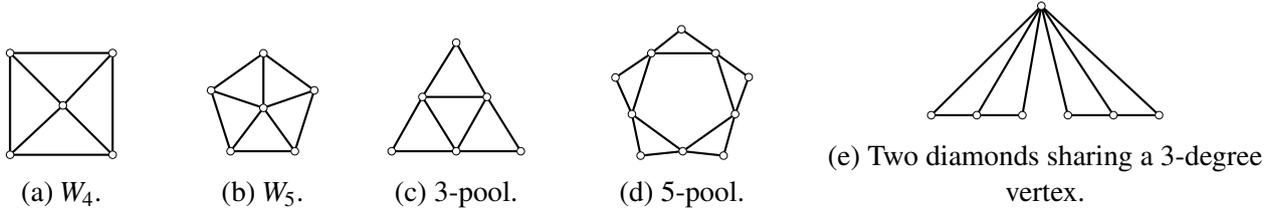
\begin{figure}[t]
		\centering
		\begin{subfigure}[b]{0.1\textwidth}
			\centering
			\begin{tikzpicture}[inner sep=1pt, outer sep=0pt]
			\node[main node] (c1) [label=93:$ $] {};
			\node[main node] (c2) [label=87:$ $, right = 1.25cm of c1] {};
			\node[main node] (c3) [label=-3:$ $, below = 1.25cm of c2] {};
			\node[main node] (c4) [label=183:$ $, left = 1.25cm and 1.25cm of c3] {};
			\node[main node] (u) [label=right:$ $, below right = 0.625cm and 0.625cm of c1] {};
			
			\path[draw,thick]
			(u) edge node {} (c1)
			(u) edge node {} (c2)
			(u) edge node {} (c3)
			(u) edge node {} (c4)
			(c1) edge node {} (c2)
			(c2) edge node {} (c3)
			(c3) edge node {} (c4)
			(c4) edge node {} (c1);
			\end{tikzpicture}
			\caption{$W_4$.}
			\label{fig:W_4}
		\end{subfigure} \qquad
		\begin{subfigure}[b]{0.1\textwidth}
			\centering
			\begin{turn}{180}
				\begin{tikzpicture}[scale=0.8, y=0.35pt, x=0.35pt, yscale=-1.000000, xscale=1.000000, inner sep=1pt, outer sep=0pt]
				\node[draw,fill=white,circle,very thin,auto=left,inner sep=1pt] (c1) at (316.8922,283.0575) [label=93:$ $] {};
				\node[draw,fill=white,circle,very thin,auto=left,inner sep=1pt] (c2) at (402.8922,283.0575) [label=87:$ $] {};
				\node[draw,fill=white,circle,very thin,auto=left,inner sep=1pt] (c3) at (428.8922,365.0575) [label=0:$ $] {};
				\node[draw,fill=white,circle,very thin,auto=left,inner sep=1pt] (c4) at (358.8922,415.0575) [label=270:$ $] {};
				\node[draw,fill=white,circle,very thin,auto=left,inner sep=1pt] (c5) at (290.8922,365.0575) [label=180:$ $] {};
				
				\node[main node] (u) [label=3:$ $, above = 0.625cm of c4] {};
				
				\path[draw,thick]
				(u) edge node {} (c1)
				(u) edge node {} (c2)
				(u) edge node {} (c3)
				(u) edge node {} (c4)
				(u) edge node {} (c5)
				(c1) edge node {} (c2)
				(c2) edge node {} (c3)
				(c3) edge node {} (c4)
				(c4) edge node {} (c5)
				(c5) edge node {} (c1);
				\end{tikzpicture}
			\end{turn}
			\caption{$W_5$.}
			\label{fig:W_5}
		\end{subfigure}\quad
		\begin{subfigure}[b]{0.15\textwidth}
			\centering
			\begin{tikzpicture}[rotate=180, inner sep=1pt, outer sep=0pt]
			\node[main node] (c1) [label=93:$ $] {};
			\node[main node] (c2) [label=87:$ $, right = 0.75cm of c1] {};
			\node[main node] (c3) [label=below:$ $, below right = 0.649519053cm and 0.375cm of c1] {};
			\node[main node] (b1) [label=above:$ $, above right = 0.649519053cm and 0.375cm of c1] {};
			\node[main node] (b2) [label=-3:$ $, right = 0.75cm of c3] {};
			\node[main node] (b3) [label=183:$ $, left = 0.75cm of c3] {};
			
			\path[draw,thick]
			(c1) edge node {} (c2)
			(c2) edge node {} (c3)
			(c3) edge node {} (c1)
			(c1) edge node {} (b1)
			(c2) edge node {} (b1)
			(c2) edge node {} (b2)
			(c3) edge node {} (b2)
			(c3) edge node {} (b3)
			(c1) edge node {} (b3);
			\end{tikzpicture}
			\caption{$3$-pool.}
			\label{fig:3-pool}
		\end{subfigure}\quad 
		\begin{subfigure}[b]{0.15\textwidth}
			\centering
			\begin{tikzpicture}[scale=0.8, y=0.35pt, x=0.35pt, yscale=-1.000000, xscale=1.000000, inner sep=1pt, outer sep=0pt]
			\node[draw,fill=white,circle,very thin,auto=left,inner sep=1pt] (c1) at (316.8922,283.0575) [label=93:$ $] {};
			\node[draw,fill=white,circle,very thin,auto=left,inner sep=1pt] (c2) at (402.8922,283.0575) [label=87:$ $] {};
			\node[draw,fill=white,circle,very thin,auto=left,inner sep=1pt] (c3) at (428.8922,365.0575) [label=0:$ $] {};
			\node[draw,fill=white,circle,very thin,auto=left,inner sep=1pt] (c4) at (358.8922,415.0575) [label=270:$ $] {};
			\node[draw,fill=white,circle,very thin,auto=left,inner sep=1pt] (c5) at (290.8922,365.0575) [label=180:$ $] {};
			\node[draw,fill=white,circle,very thin,auto=left,inner sep=1pt] (b1) at (357.1657,250.8029) [label=90:$ $] {};
			\node[draw,fill=white,circle,very thin,auto=left,inner sep=1pt] (b2) at (447.0006,315.2609) [label=87:$ $] {};
			\node[draw,fill=white,circle,very thin,auto=left,inner sep=1pt] (b3) at (413.4775,421.0017) [label=0:$ $] {};
			\node[draw,fill=white,circle,very thin,auto=left,inner sep=1pt] (b4) at (302.5624,421.0027) [label=180:$ $] {};
			\node[draw,fill=white,circle,very thin,auto=left,inner sep=1pt] (b5) at (268.8974,315.7932) [label=93:$ $] {};
			
			\foreach \from/\to in {c1/c2, c2/c3, c3/c4, c4/c5, c5/c1, c1/b1, b1/c2, b2/c2, b2/c3, b3/c3, b3/c4, b4/c4, b4/c5, b5/c5, b5/c1}
			\draw[thick] (\from) -- (\to);
			\end{tikzpicture}
			\caption{5-pool.}
			\label{fig:5-pool}
		\end{subfigure}	\quad
		\begin{subfigure}[b]{0.35\textwidth}
			\centering
			\begin{tikzpicture}[rotate=180, inner sep=1pt, outer sep=0pt]
			\node[main node] (u) [label=above:$ $] {};
			\node[main node] (c1) [label=below:$ $, below left = 1.95cm of u] {};
			\node[main node] (c2) [label=below:$ $, right = 0.5cm of c1] {};
			\node[main node] (c3) [label=below:$ $, right = 0.5cm of c2] {};
			\node[main node] (c4) [label=below:$ $, right = 0.5cm of c3] {};
			\node[main node] (c5) [label=below:$ $, right = 0.5cm of c4] {};
			\node[main node] (c6) [label=below:$ $, right = 0.5cm of c5] {};
			
			\path[draw,thick]
			(u) edge node {} (c1)
			(u) edge node {} (c2)
			(u) edge node {} (c3)
			(c1) edge node {} (c2)
			(c2) edge node {} (c3)
			(u) edge node {} (c4)
			(u) edge node {} (c5)
			(u) edge node {} (c6)
			(c4) edge node {} (c5)
			(c5) edge node {} (c6);
			\end{tikzpicture}
			\caption{Two diamonds sharing a $3$-degree vertex.}
			\label{fig:bidiamond}
		\end{subfigure}	
		\caption{Some examples of forbidden subgraphs.}
		\label{fig:forb_graphs}
	\end{figure}
	
	Since~$\mathcal{BM}$ is closed under taking subgraphs, we can search for families of minimal forbidden subgraphs, that is, graphs that cannot be in~$\mathcal{BM}$.
	Figure~\ref{fig:forb_graphs} depicts some examples of them,	while Lemma~\ref{lemma1} collects some properties of graphs in~$\mathcal{BM}$.\\
	\begin{lemma}\label{lemma1}
		For a graph~$G \in \mathcal{BM}$ and a bipartizing matching~$M$ of~$G$, the following statements hold:
		\begin{itemize}
			\item[\textnormal{(i)}] For every diamond~$D$ of~$G$, $M$ matches both vertices of degree~3 of~$D$.
			\item[\textnormal{(ii)}] For every~$v \in V(G)$, $G[N_G(v)]$ does not contain two disjoint~$P_3$.
			\item[\textnormal{(iii)}] $G$ does not contain a~$W_k$ as a subgraph, for all~$k\geq 4$.
			\item[\textnormal{(iv)}] $G$ does not contain an odd~$k$-pool as subgraph, for all odd~$k \geq 3$.
		\end{itemize}
	\end{lemma}
	\begin{proof}
		(i) Let~$D$ be a diamond subgraph of~$G$, such that~$V(D) = \{u, v_1, v_2, v_3\}$ and~$d_D(u)=d_D(v_2)=3$.
		We can see that~$M \cap E(D)$ must be one of the following sets: $\{uv_1, v_2v_3\}$, $\{v_1v_2, uv_3\}$, $\{uv_2\}$.
		In each one, $u$ and~$v_2$ are matched by~$M$.
		
		(ii) Suppose for a contradiction that~$v \in V(G)$ is such that~$G[N_G(v)]$ contains two disjoint~$P_3$: $P$ and~$P^{\prime}$.
		This implies that~$G[\{v\}\cup P]$ and~$G[\{v\}\cup P^{\prime}]$ are diamonds sharing a vertex of degree three.
		Then, by Statement (i) it follows that~$M$ contains two incident edges, one for each diamond, a contradiction.
		
		(iii) Suppose for a contradiction that~$G$ contains a subgraph~$H$ isomorphic to~$W_k$, $k\geq 4$, where~$u$ is universal in~$H$.
		If~$k \geq 7$, it follows that~$u$ has two disjoint~$P_3$ in its neighborhood, a contradiction by Statement (ii).
		Then~$k \leq 6$ and it can be easily verified that~$W_4$ and~$W_5$ are not in~$\mathcal{BM}$.
		
		(iv) Suppose for a contradiction that~$G$ contains a subgraph~$H$ isomorphic to an odd~$k$-pool, for an odd~$k\geq 3$.
		Let~$C = p_1p_2 \ldots p_kp_1$ be its internal cycle and let $B = \{b_1, b_2, \dotsc, b_k\}$ be the set of vertices of the borders of~$H$, such that~$\{p_ib_i, p_{i+1}b_i\} \subset E(H)$, for all~$1 \leq i \leq k$ modulo~$k$.
		Clearly~$M$ must contain some edge of~$C$ and one edge of every triangle~$p_ib_ip_{i+1}p_i$.
		W.l.o.g., consider~$p_1p_2 \in M\cap E(C)$.
		This implies that~$M$ contains no edge in~$\{p_2b_2, p_2p_3, p_1b_k, p_1p_k\}$.
		Therefore, $p_kb_k$ and~$p_3b_2$ must be in~$M$, which forbids two more edges from the triangles~$p_{k-1}b_{k-1}p_kp_{k-1}$ and~$p_3b_3p_4p_3$.
		Continuing this process, it follows that~$c_{\left\lfloor\frac{k+3}{2}\right\rfloor}$, which is at the same distance of~$p_1$ and~$p_2$ in~$C$, must contain two incident edges in~$M$, a contradiction.
	\end{proof}		
	Clearly every graph in~$\mathcal{BM}$ admits a proper $4$-coloring.
	Hence every graph in~$\mathcal{BM}$ is~$K_5$-free, implying that all graphs in~$\mathcal{BM}$ have maximum clique bounded by~4.
	More precisely, every graph in~$\mathcal{BM}$ is~$W_5$-free, as in Lemma~\ref{lemma1}, which implies that even some properly~$3$-colorable graphs do not admit a decycling matching.
	
	\subsection{Basics on Parameterized Complexity}
	\label{subsec:paramComp}
	
	For a basic background on parameterized complexity, we refer to~\cite{CyganFKLMPPS15,DF13}, and only some definitions are given here.
	A \emph{parameterized problem} is a language~$L \subseteq \Sigma^* \times \mathbb{N}$.
	For an instance~$I=(x,k) \in \Sigma^* \times \mathbb{N}$, $k$ is called the \emph{parameter}.
	A parameterized problem is \emph{fixed-parameter tractable} ({\sf FPT}) if there exists an algorithm~$\mathsf{A}$, a computable function~$f$, and a constant~$c$ such that given an instance $I=(x,k)$, $\mathsf{A}$ (called an {\sf FPT} \emph{algorithm}) correctly decides whether~$I \in L$ in time bounded by~$f(k) \cdot |I|^c$.
	
	A fundamental concept in parameterized complexity is that of \emph{kernelization}.
	A kernelization algorithm, or just \emph{kernel}, for a parameterized problem~$\Pi $ takes an instance~$(x,k)$ of the problem and, in time polynomial in~$|x| + k$, outputs an instance~$(x',k')$ such that~$|x'|, k' \leqslant g(k)$ for some function~$g$, and $(x,k) \in \Pi$ if and only if~$(x',k') \in \Pi$.
	The function~$g$ is called the \emph{size} of the kernel and may be viewed as a measure of the ``compressibility'' of a problem using polynomial-time preprocessing rules.
	A kernel is called \emph{polynomial} (resp. \emph{linear}) if~$g(k)$ is a polynomial (resp. linear) function in $k$.
	It is nowadays a well-known result in the area that a decidable problem is in~{\sf FPT} if and only if it has a kernelization algorithm.
	However, the kernel that one obtains in this way is typically of size at least exponential in the parameter.
	A natural problem in this context is to find polynomial or linear kernels for problems that are in~{\sf FPT}.

	\section{{\sf NP}-Completeness}
	\label{sec:NPc}
	
	Let~$F$ be a Boolean formula in~$3$-CNF such that~$\textbf{X} = \{X_1, X_2, \dotsc, X_n\}$ is the variable set and~$\textbf{C} = \{C_1, C_2, \dotsc, C_m\}$ is the clause set of~$F$.
	We represent the positive and negative literals of~$X_i$ as~$x_i$ and~$\overline{x_i}$, respectively.
	The \textit{associated graph}~$G_F=(V, E)$ of~$F$ is the bipartite graph such that there exists a vertex for every variable and clause of~$F$, where~$(\textbf{X}, \textbf{C})$ is the bipartition of~$V(G_F)$, and there exists an edge~$X_iC_j \in E(G_F)$ if and only if~$C_j$ contains either~$x_i$ or~$\overline{x_i}$.
	We say that~$F$ is a \textit{planar formula} if~$G_F$ is planar.
	
	In order to prove Theorem~\ref{thm:np-planar4}, we first present a polynomial-time reduction from the well known problem {\sc Positive Planar 1-In-3-SAT}~\cite{mr} to {\sc Planar 1-In-3-SAT$_3$}, defined as follows.
	
	\vspace{-.05cm}
	\probld	
	{Positive Planar 1-In-3-SAT{\normalfont ~\cite{mr}}}
	{A planar formula~$F$ in~$3$-CNF with no negative literals.}
	{Is there a truth assignment to the variables of~$F$, in which each clause has exactly one literal assigned true?}
	
	\vspace{-.05cm}
	\probld
	{Planar 1-In-3-SAT$_3$}
	{A planar formula~$F$ in~CNF, where each clause has either~$2$ or~$3$ literals and each variable occurs at most~$3$ times.
		Moreover, each positive literal occurs at most twice, while every negative literal occurs at most once in~$F$.}
	{Is there a truth assignment to the variables of~$F$ in which each clause has exactly one true literal?}

	\begin{theorem}\label{thm:NPplanar1-3}
		{\sc Planar 1-In-3-SAT$_3$} is~{\sf NP}-complete.
	\end{theorem}
	\begin{proof} Let~$F$ be a planar formula in~$3$-CNF such that~$\textbf{X} = \{X_1, X_2, \dotsc, X_n\}$ is the variable set
		and~$\textbf{C} = \{C_1, C_2, \dotsc, C_m\}$ is the clause set of~$F$.
		Since verifying whether a graph is planar can be done in linear time~\cite{ht74}, as well as whether
		a formula in~$3$-CNF has a truth assignment, the problem is in~{\sf NP}.
		
		We construct a formula~$F^\prime$ from~$F$ as follows.
		If~$X_i$ is such that~$d_{G_F}(X_i) = k \geq 3$, then we create~$k$ new variables~$X_i^z$
		replacing the~$j^{th}$ ($1 \leq j \leq k$) occurrence of~$X_i$ by a variable~$X_i^j$,
		where a literal~$x_i$ (resp.~$\overline{x_i}$) is replaced by a literal~$x_i^j$ $(\textnormal{resp. } \overline{x_i^j})$.
		In addition, we create~$k$ new clauses~$C_i^j = \left(x_i^{j}, \overline{x_i^{j+1}}\right)$, 
		for~$j \in \{1, \dotsc, k-1\}$, and~$C_i^k = \left(x_i^{k}, \overline{x_i^1}\right)$.
		
		Let~$S$ be the set of all vertices~$X_i \in V(G_F[\textbf{X}])$ with~$d_{G_F}(X_i) = k \geq 3$.
		For such a vertex~$X_i \in S$, let~$X'_i = \{X_i^1, \dotsc, X_i^{k}\}$ and~$C'_i = \{C_i^1, \dotsc, C_i^k\}$.
		
		Note that, the associated graph~$G_{F^\prime}$ can be obtained from~$G_F$
		by replacing the corresponding vertex of~$X_i\in S$ by a cycle of length~$2d_{G_F}(X_i)$
		induced by the corresponding vertices of the new clauses in~$C'_i$ and the new variables in~$X'_i$.
		In addition, for each~$X_i \in S$ and~$C_j \in N_{G_F}(X_i)$, an edge~$X_i^tC_j$ is added
		in~$E(G_{F^\prime})$, such that every corresponding vertex~$X_i^t\in X'_i$ has exactly one neighbor~$C_j \notin C'_i$.
		
		As we can see, every variable~$X$ occurs at most~$3$ times in the clauses of~$F^\prime$,
		since every variable~$X_i$ with~$d_{G_F}(X_i)\geq 3$ is replaced by~$d_{G_{F}}(X_i)$ new
		variables that are in exactly~$3$ clauses of~$F^\prime$.
		By the construction, each literal occurs at most twice.
		Moreover, if~$F$ has no negative literals, then only the new variables have a negative literal
		and each one occurs exactly once in~$F^{\prime}$.
		
		Consider a planar embedding~$\Psi$ of~$G_{F}$.
		We construct~$G_{F^{\prime}}$ replacing each corresponding vertex~$X_i \in S$ by a cycle of length~$2d_{G_{F}}(X_i)$,
		as described above.
		After that, in order to preserve the planarity, we can follow the planar embedding~$\Psi$ to
		add a matching between vertices corresponding to variables in such a cycle and vertices corresponding
		to clauses~$C_j \notin C'_i$ and that~$X_i \in C_j$.
		This matching indicates in which clause of~$C'_i$ a given new variable will replace~$X_i$ in~$F^{\prime}$.
		Thus, without loss of generality, if~$G_{F}$ is planar, then we can assume that~$F^{\prime}$ is planar as well.
		
		As we can observe, for any truth assignment of~$F^{\prime}$, all $X_i^t\in X_i$ (for a given variable~$X_i$ of~$F$)
		have the same value.
		Hence, any clause of~$F^{\prime}$ with exactly two literals has true and false values.
		At this point, it is easy to see that~$F$ has a 1-in-3 truth assignment if and only if~$F^{\prime}$ has a 1-in-3
		truth assignment.
	\end{proof}
	
	Now we show the {\sf NP}-completeness of {\sc BM}.
	Let us call the graph depicted in~Figure~\ref{fig:head} by \textit{head}.
	We also call the vertex~$v$ as the \textit{neck} of the head.
	Given a graph~$G$, the next lemma shows that this structure is very useful to ensure that some edges cannot be in any odd decycling matching of~$G$.
	
	\begin{figure}[t!]
		\begin{subfigure}[b]{0.45\textwidth}
			\centering
			\includegraphics[width=0.37\textwidth]{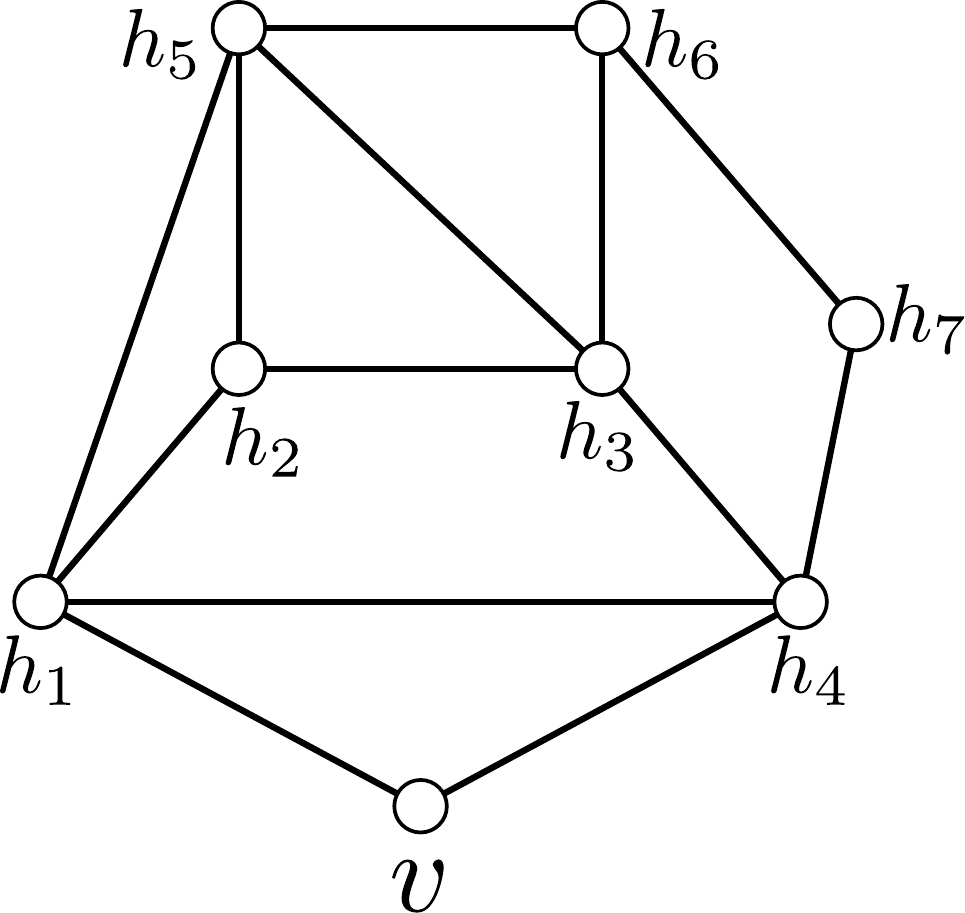}
			\caption{The head graph~$H$.}
			\label{fig:head}
		\end{subfigure} 
		\begin{subfigure}[b]{0.45\textwidth}
			\centering
			\includegraphics[width=0.37\textwidth]{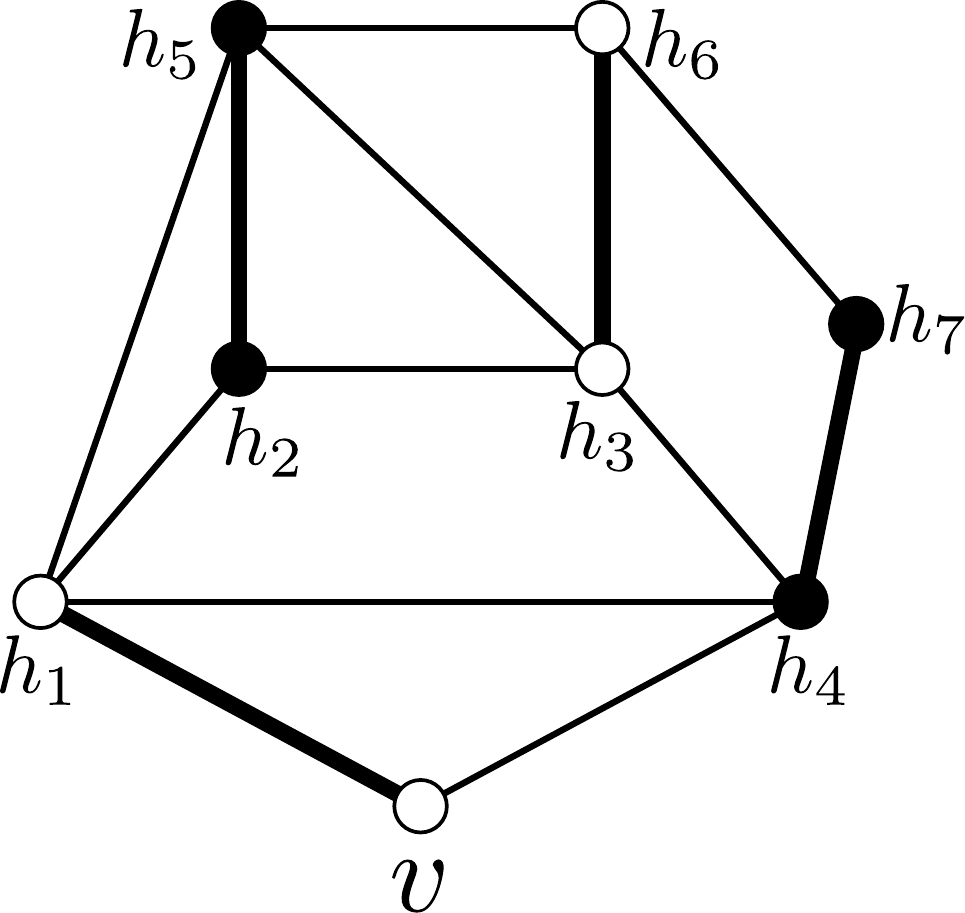}
			\caption{The bipartizing matching of~$H$.}
			\label{fig:head1}
		\end{subfigure} 
		\caption{The head graph and its unique bipartizing matching.
		}
		\label{fig:Head}
	\end{figure}
	
	\begin{lemma}\label{lem:head}
		An induced head~$H$ with neck~$v$ admits exactly one bipartizing matching~$M$.
		Moreover~$v$ is matched by~$M$.
	\end{lemma}
	\begin{proof} 
		We can see that the thicker edges in Figure~\ref{fig:head1} compose a bipartizing matching~$M$ of~$H$,
		where the color of the vertices indicates the part of each one in~$G-M$.
		Moreover~$M$ satisfies the lemma.
		
		Now let us suppose that~$H$ admits another bipartizing matching~$M'$, which does not match~$v$.
		In this case, we get that~$vh_1$ and~$vh_4$ does not belong to~$M'$, implying that~$h_1h_4 \in M'$.
		By the triangle~$h_1h_2h_5$, it follows that~$h_2h_5$ must be in~$M'$.
		Hence the cycle~$vh_1h_2h_3h_4v$ remains in~$G-M'$, a contradiction.
		
		Let~$M''$ be a bipartizing matching that contains~$vh_4$.
		In this case, the edge~$h_1h_2$ cannot be in~$M''$, otherwise the cycle~$h_1h_4h_3h_2h_5h_1$ survives in~$G-M''$.
		In the same way, the edge~$h_1h_5 \notin M''$, otherwise the cycle~$h_1h_2h_5h_3h_4h_1$ is not destroyed by~$M''$.
		Therefore we get that~$h_2h_5$ must be in~$M''$, which implies that the cycle~$h_1h_4h_3h_2h_5h_1$ belongs to~$G-M''$, a contradiction.
	\end{proof}
	
	\begin{lemma}\label{lem:gagdetPool}
		Let~$G$ be an odd~$k$-pool with internal cycle~$C = p_1p_2 \dots p_kp_1$.
		Let~$b$ be a border of~$G$, where~$N_G(b) = \{p_1, p_k\}$.
		Then every bipartizing matching~$M$ of~$G-b$ contains exactly one edge of~$C$.
		Moreover, $c_1c_k \notin M$ for such bipartizing matchings.
	\end{lemma}
	\begin{proof} Let~$b_i$ be the $i$-th-border of~$G$, such that~$N_G(b_i)=\{p_i, p_{i+1}\}$, $1 \leq i \leq k-1$.
		Since~$k$ is odd, every bipartizing matching of~$G-b$ must contain at least one edge of~$C$.
		
		First, suppose that~$G-b$ has a bipartizing matching~$M$ containing~$p_1p_k$.
		In this case, we get that the edges in~$\{p_1p_2, p_1b_1, p_kp_{k-1}, p_kb_{k-1}\}$ cannot be in~$M$.
		Thus~$M$ must contain edges~$b_1p_2$ and~$b_{k-1}p_{k-1}$.
		In the same way, we can see that~$M$ cannot contain the edges $\{p_2p_3, p_2b_2, p_{k-1}p_{k-2},$ $p_{k-1}b_{k-2}\}$.
		Hence, it can be seen that all edges incident to~$p_{\frac{k+1}{2}}$ are forbidden to be in~$M$,
		which implies that the triangles containing~$p_{\frac{k+1}{2}}$ have no edge in~$M$, a contradiction by
		the choice of~$M$.
		
		Let~$p_ip_{i+1}$ be an edge of~$C$ in a bipartizing matching~$M$ of~$G-b$, with~$i\neq k$.
		In a same fashion, the edges in~$\{p_ip_{i-1}, p_ib_{i-1}, p_{i+1}p_{i+2}, p_{i+1}b_{i+1}\}$ cannot be in~$M$.
		Following this pattern, we can see that every edge~$b_jp_{j+1}$ must be in~$M$, for
		every~$j \in \{1, \dotsc, k\} \setminus \{i\}$.
		Since~$M$ contains only one edge of~$C$ and one edge of every triangle of~$G-b$,
		it follows that~$M$ is unique, for each edge~$p_ip_{i+1}$, $i \neq k$.
	\end{proof}
	
	\noindent \textbf{{\sf NP}-Completeness for Planar Graphs of Maximum Degree~4.}
	
	We prove the~{\sf NP}-completeness by a reduction from {\sc Planar 1-In-3-SAT$_3$}.
	We first prove it for~$3$-colorable planar graphs of maximum degree~5.
	The circles with an~$H$ in the figures are induced head graphs, whose neck is the vertex
	touching the circle.
	
	\begin{figure}[t]
		\centering
		\begin{subfigure}[b]{0.27\textwidth}
			\centering
			\includegraphics[width=\textwidth]{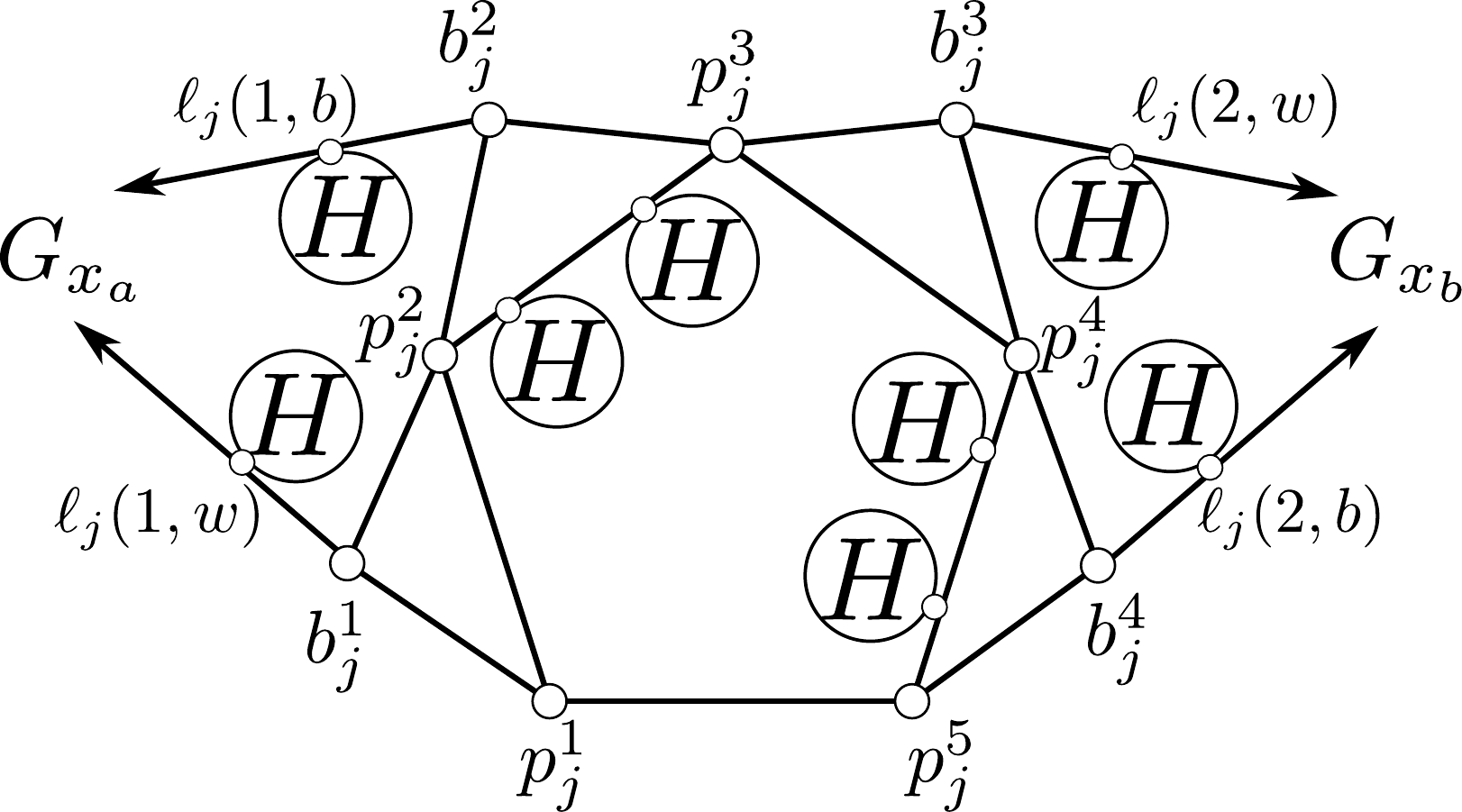}
			\caption{For clauses of size two.}
			\label{fig:ClGadget2}
		\end{subfigure}
		\begin{subfigure}[b]{0.37\textwidth}
			\centering
			\includegraphics[width=\textwidth]{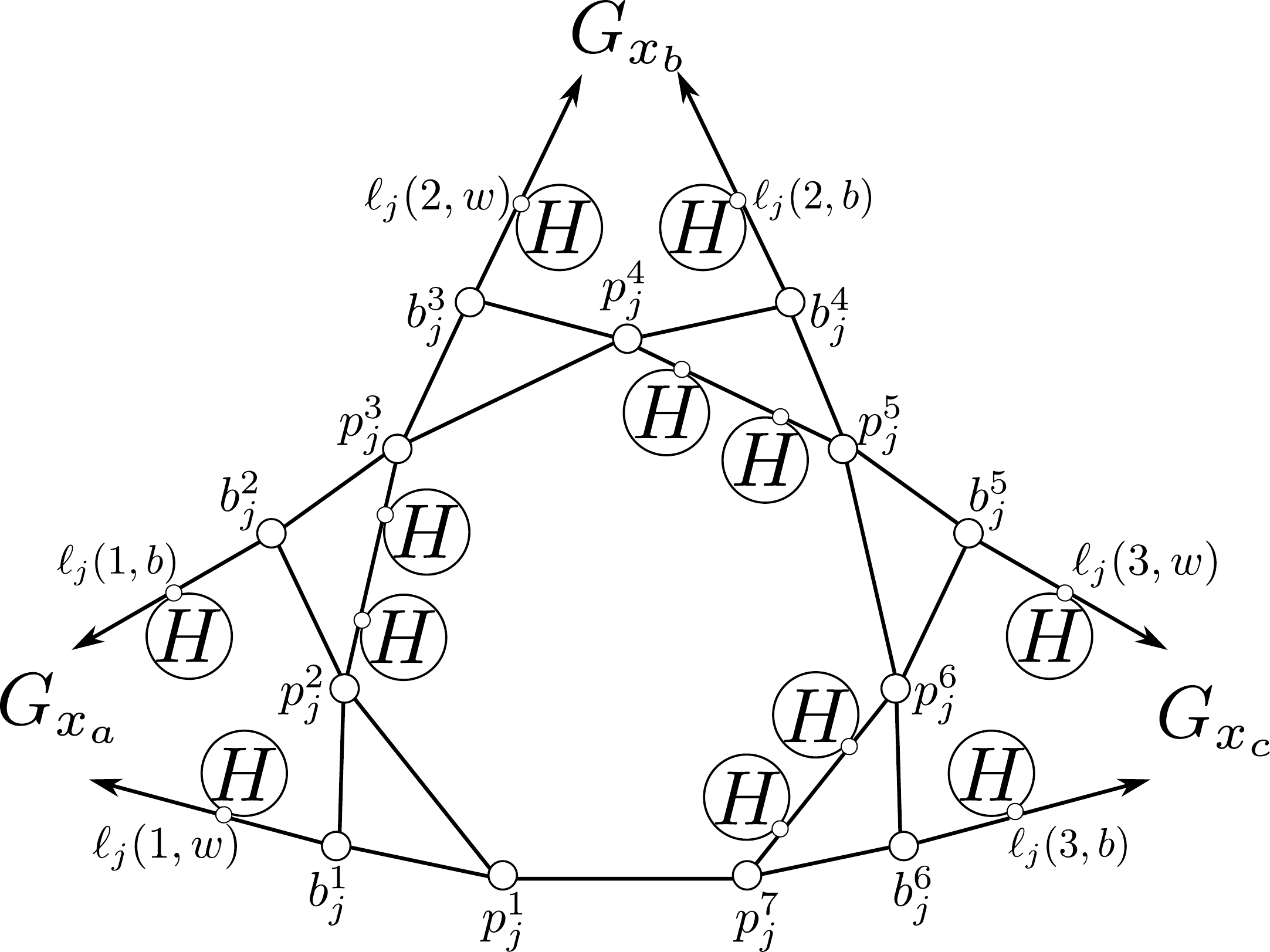}
			\caption{For clauses of size three.}
			\label{fig:ClGadget3}
		\end{subfigure}
		\begin{subfigure}[b]{0.34\textwidth}
			\centering
			\includegraphics[width=\textwidth]{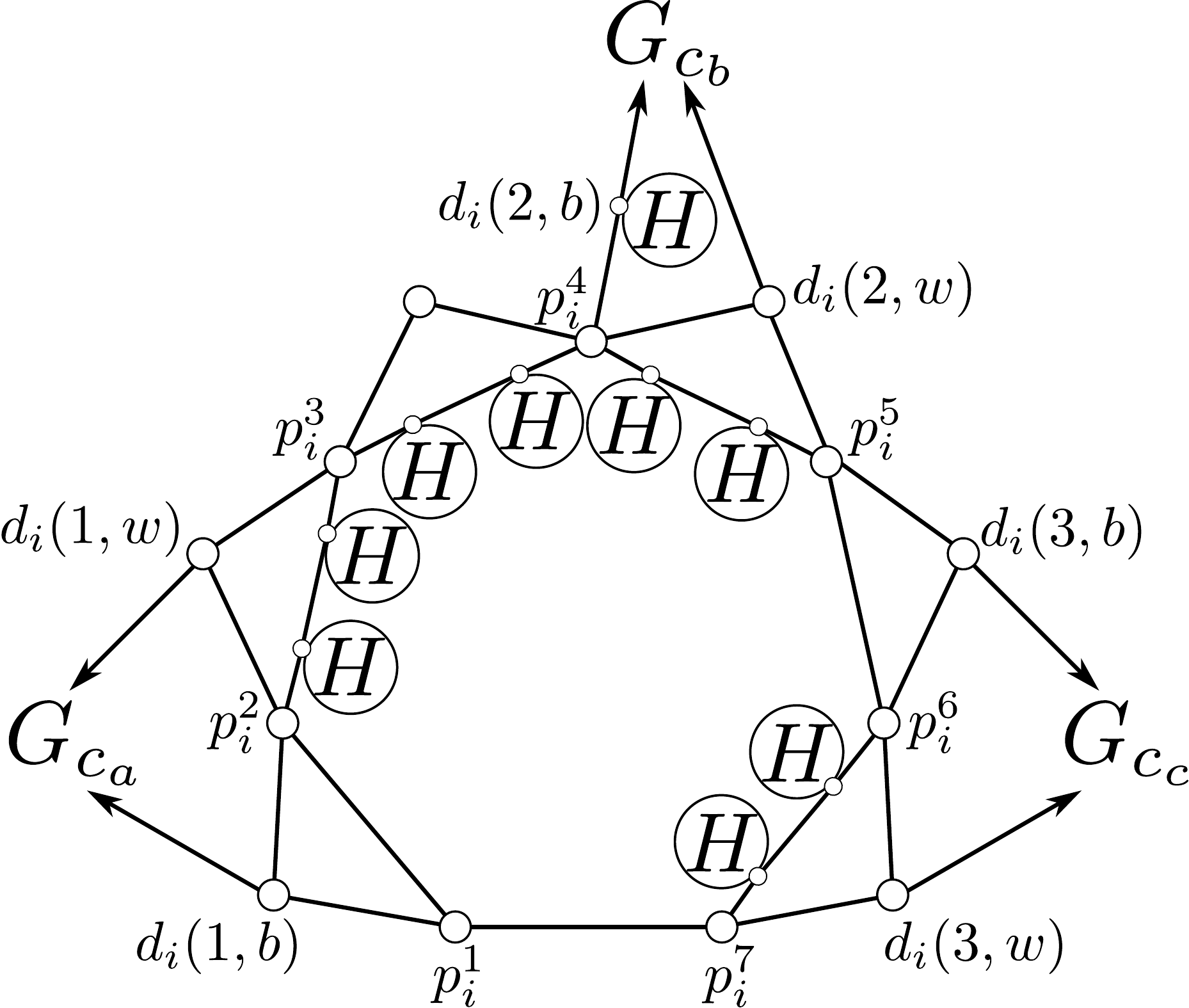}
			\caption{Variable gadget.}
			\label{fig:VarGadgetNP5}
		\end{subfigure}
		\caption{Figure~\ref{fig:ClGadget2} and~\ref{fig:ClGadget3} are the clause gadgets and Figure~\ref{fig:VarGadgetNP5}
			is the variable gadget in Theorem~\ref{thm:NPCompDeg5}.
			Each pair of arrow edges connects~$G_{X_i}$ to one clause gadget~$G_{C}$ such that~$X_i \in C$.}
		\label{fig:clausesDeg5}
	\end{figure}	
	
	\begin{theorem}\label{thm:NPCompDeg5}
		{\sc BM} is~{\sf NP}-complete for~$3$-colorable planar graphs~$G$ of~$\Delta(G)= 5$.
	\end{theorem}
	\begin{proof} Let~$F$ be an instance of {\sc Planar 1-In-3-SAT$_3$}, with~$X=\{X_1, X_2, \dotsc, X_n\}$
		and $C = \{C_1, C_2, \dotsc,$ $C_m\}$ be the sets of variables and clauses of~$F$, respectively.
		We construct a planar graph~$G=(V, E)$ of maximum degree~$5$ as follows:
		\begin{itemize}
			\item For each clause~$C_j \in C$, we construct a gadget~$G_{C_j}$ as depicted in Fig.~\ref{fig:ClGadget2} and Fig.~\ref{fig:ClGadget3}.
			Such gadgets are just a~$5$-pool and a~$7$-pool that we remove a border vertex, for clauses of size~$2$ and~$3$, respectively.
			Moreover, for the alternate edges of the internal cycle we subdivide them twice and append a head to each such a new vertex.
			Finally, we add two vertices~$\ell_j(k, w)$ and~~$\ell_j(k, b)$, such that~$b_j^{2k-1}\ell_j(k, w) \in E(G)$
			and~$b_j^{2k}\ell_j(k, b) \in E(G)$, for~$k \in \{1,2,3\}$.
			We append a head to all such new vertices.
			\item For each variable~$X_i \in X$, we construct a gadget~$G_{X_i}$ as depicted in Figure~\ref{fig:VarGadgetNP5}.
			This gadget is a~$7$-pool that we remove a border vertex.
			Moreover, we subdivide twice~the edges~$p_i^2p_i^3$, $p_i^3p_i^4$, $p_i^4p_i^5$, and~$p_i^6p_i^7$,
			appending a head to each new vertex.~We rename each border vertex~$b_i^{2k-1}$ as~$d_i(k, b)$, $k \in\{1, 3\}$,
			and~$b_i^{2k}$ as~$d_i(k, w)$, $k \in \{1,2,3\}$.
			Finally, a vertex~$d_i(2, b)$ with a pendant head and adjacent to~$p_i^4$ is added.
			\item The connection between clause and variable gadgets is as in Figure~\ref{fig:clausesDeg5}.
			Each pair of arrow edges in a variable gadget~$G_{X_i}$ corresponds to the same pair in a clause
			gadget~$G_{C_j}$, where~$X_i \in C_j$.
			Precisely, if~$x_i \in C_j$, then we add the edges $\ell_j(k, b)d_i(k^{\prime}, b)$
			and $\ell_j(k, w)d_i(k^{\prime}, w)$, for some~$k \in \{1, 2, 3\}$ and some~$k^{\prime} \in \{1, 2\}$.
			Note that each $\ell_j(k, b)$ (and $\ell_j(k, w)$) represents precisely one variable and it is connected to only
			one of~$d_i(k', b)$ (and $d_i(k',w)$) in~$G_{X_i}$.
			However, if~$\overline{x_i} \in C_j$, then we add the edges $\ell_j(k, b)d_i(3, b)$ and $\ell_j(k, w)d_i(3, w)$,
			for some~$k \in \{1, 2, 3\}$.
			Note that~$d_i(3, b)$ and $d_i(3, w)$ represent~$\overline{x_i}$.
			\item For a variable occurring exactly twice in~$F$, just consider those connections corresponding to literals
			of~$X_i$ in the clauses of~$F$, i.e., the pair~$d_i(3, b)$, $d_i(3, w)$ (resp. $d_i(2, b)$, $d_i(2, w)$) will be used
			to connect to a clause gadget only if $\overline{x_i}$ occurs (resp. does not occur) in~$F$.
		\end{itemize}					
		Let~$G$ be the graph obtained from~$F$ by the above construction.
		We can see that~$G$ has maximum degree~$5$, where the only vertices with degree~$5$ are those~$p_i^4$,
		for each variable gadget.
		Furthermore, it is clear that~$G$ is~$3$-colorable.
		
		It remains to show that if~$G_F$ is planar, then~$G$ is planar.
		Consider a planar embedding~$\psi$ of~$G_F$.
		We replace each vertex~$v_{X_i}$ of~$G_F$ by a variable gadget~$G_{X_i}$, as well as every vertex~$v_{C_j}$
		of~$G_{F}$ by a clause gadget~$G_{C_j}$.
		The clause gadgets correspond to clauses of length two or three, which depends on the degree of~$v_{C_j}$ in~$G_F$.
		Since the clause and variable gadgets are planar, we just need to show that the	connections among them keep	the planarity.
		Given an edge~$v_{X_i}v_{C_j} \in E(G_F)$, we connect~$G_{X_i}$ and~$G_{C_j}$ by duplicating such an edge as parallel
		edges~$\ell_j(k, w)d_i(k^{\prime}, w)$ and~$\ell_j(k, b)d_i(k^{\prime}, b)$, for some~$k \in \{1, 2, 3\}$
		and some~$k^{\prime} \in \{1, 2\}$, or~$\ell_j(k, b)d_i(3, b)$ and~$\ell_j(k, w)d_i(3, w)$, for some~$k \in \{1, 2, 3\}$,
		as previously discussed.
		Hence~$G$ is also planar.
		
		\begin{figure}[t]
			\begin{subfigure}[b]{.33\textwidth}
				\centering
				\includegraphics[width=\textwidth]{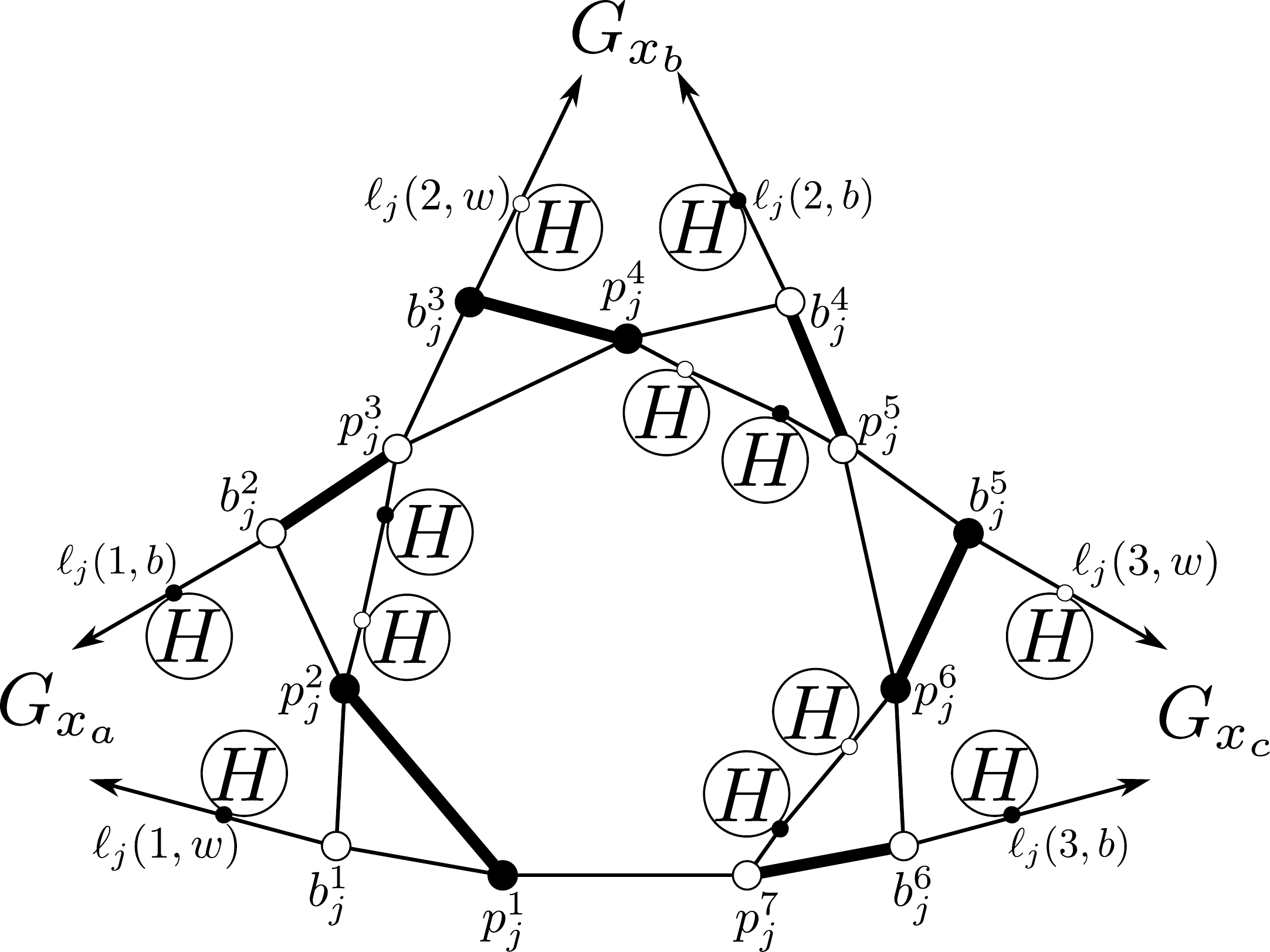}
				\caption{$p_j^1p_j^2 \in M$}
				\label{fig:ClGadget3Mathing1}
			\end{subfigure} 
			\begin{subfigure}[b]{.33\textwidth}
				\centering
				\includegraphics[width=\textwidth]{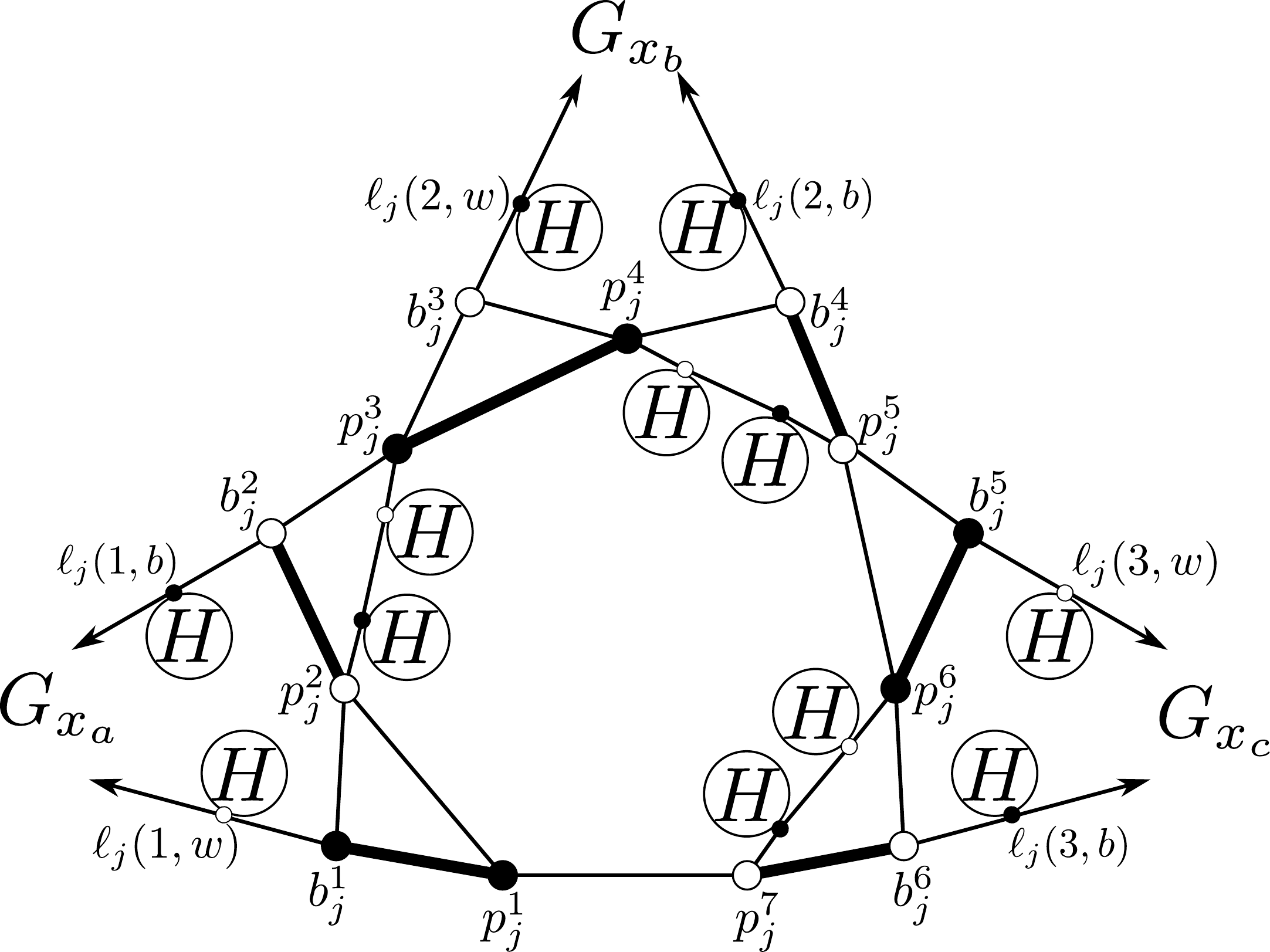}
				\caption{$p_j^3p_j^4 \in M$}
				\label{fig:ClGadget3Mathing2}
			\end{subfigure}
			\begin{subfigure}[b]{.33\textwidth}
				\centering
				\includegraphics[width=\textwidth]{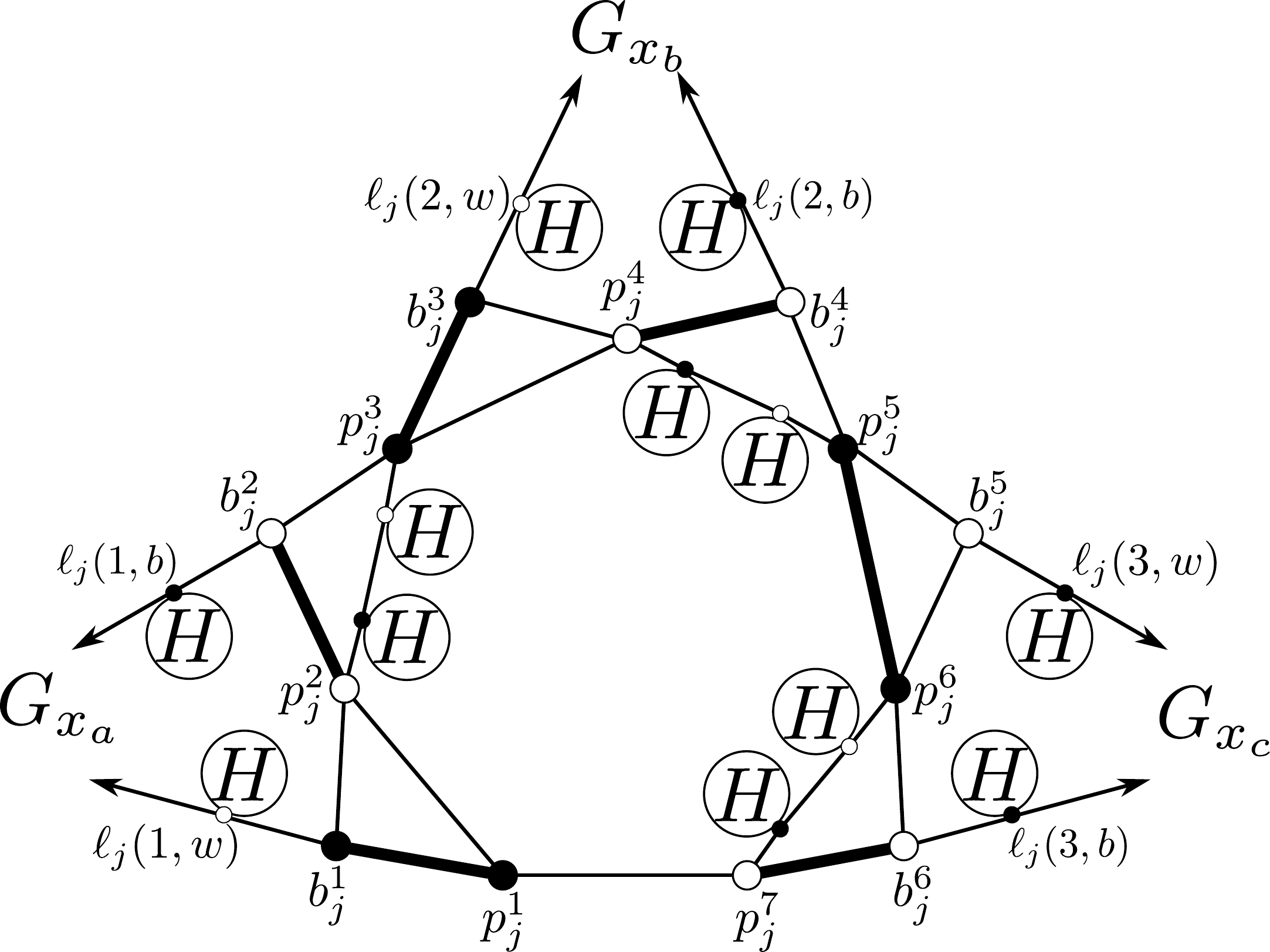}
				\caption{$p_j^5p_j^6 \in M$}
				\label{fig:ClGadget3Mathing3}
			\end{subfigure}
			\caption{Configurations by removing a bipartizing matching clause gadget of size three.}
			\label{fig:ClGadget3Mathing}
		\end{figure}
		
		In order to prove that~$F$ is satisfiable if and only if~$G \in \mathcal{BM}$, we discuss some considerations
		related to bipartizing matchings of the clause and variable gadgets.
		By Lemma~\ref{lem:gagdetPool}, we know that the graph obtained by removing a border from an odd~$k$-pool admits a unique
		bipartizing	matching for each edge of the internal cycle, except that whose endvertices are adjacent to the removed border.
		Furthermore, Lemma~\ref{lem:head} implies that each external edge incident to the neck of an induced head
		cannot be in any bipartizing matching of~$G$.
		Figure~\ref{fig:ClGadget3Mathing} shows the possible bipartizing matchings~$M$, stressed edges,
		for the clause gadget~$G_{C_j}$	of clauses of size three.
		The black and white vertex assignment represents the bipartition of~$G_{C_j}-M$.
		Exactly one pair of vertices~$\ell_j(k, w)$ and~$\ell_j(k, b)$ ($k \in \{1, 2, 3\}$) is such that
		they have the same color, while the others pairs have opposite colors.
		Precisely, we can see that~$\ell_j(k, w)$ has the same color for each pair with opposite color vertices
		as well	as~$\ell_j(k, b)$, for each bipartizing matching of~$G_{C_j}$.
		Hence, we can associate one literal~$x_j^1$, $x_j^2$, and~$x_j^3$ to each pair of
		vertices~$\ell_j(k, w)$ and~$\ell_j(k, b)$, $k \in \{1, 2, 3\}$.
		A similar analysis can be done for clause gadgets of clauses of size two.
		
		In the same way, each variable gadget~$G_{X_i}$ admits two possible bipartizing	matchings~$M$,
		as depicted in Figure~\ref{fig:VarGadgetMathing}.
		We can see that the pair~$d_i(3,b)$ and~$d_i(3,w)$ has a different assignment for the other two
		pairs~$d_i(k,b)$ and~$d_i(k,w)$, $k \in \{1, 2\}$.
		Moreover, the last two pairs have the same assignment, as depicted in Figure~\ref{fig:VarGadgetMathing1}
		and Figure~\ref{fig:VarGadgetMathing2}.
		One more detail is that the unique possibilities for such pairs is that~$d_i(3,b)$ and~$d_i(3,w)$
		have opposite assignments if and only if the vertices~$d_i(k,b)$ and~$d_i(k,w)$ have the same assignment, $k \in \{1, 2\}$.
		Therefore we can associate the positive literal~$x_i$ to the pairs~$d_i(k,b)$ and~$d_i(k,w)$, $k \in \{1, 2\}$,
		while~$\overline{x_i}$ can be represented by~$d_i(3,b)$ and~$d_i(3,w)$.
		
		As observed above for clause gadgets, we can associate true value to the pair of vertices~$\ell_j(k, w)$
		and~$\ell_j(k, b)$ with same color, $k \in \{1, 2, 3\}$.
		Hence exactly one of them is true, that is, exactly one literal of~$C_j$ is true.
		Moreover, each variable gadget has two positive literals and one negative.
		Hence, if~$G \in \mathcal{BM}$, then every clause gadget has exactly one true literal and every variable has
		a correct truth assignment, which implies that~$F$ is satisfiable.
		
		Conversely, if~$F$ is satisfiable, then each clause has exactly one true literal.
		Thus, for each clause gadget~$G_{C_j}$ we associate a same color to the pair of vertices corresponding to its true literal.
		By Figure~\ref{fig:ClGadget3Mathing}, there is an appropriate choice of a bipartizing matching for each true literal of~$C_j$.
		The same holds for each variable gadget.
	\end{proof}
	
	\begin{figure}[t]
		\begin{subfigure}[b]{.49\textwidth}
			\centering
			\includegraphics[width=.68\textwidth]{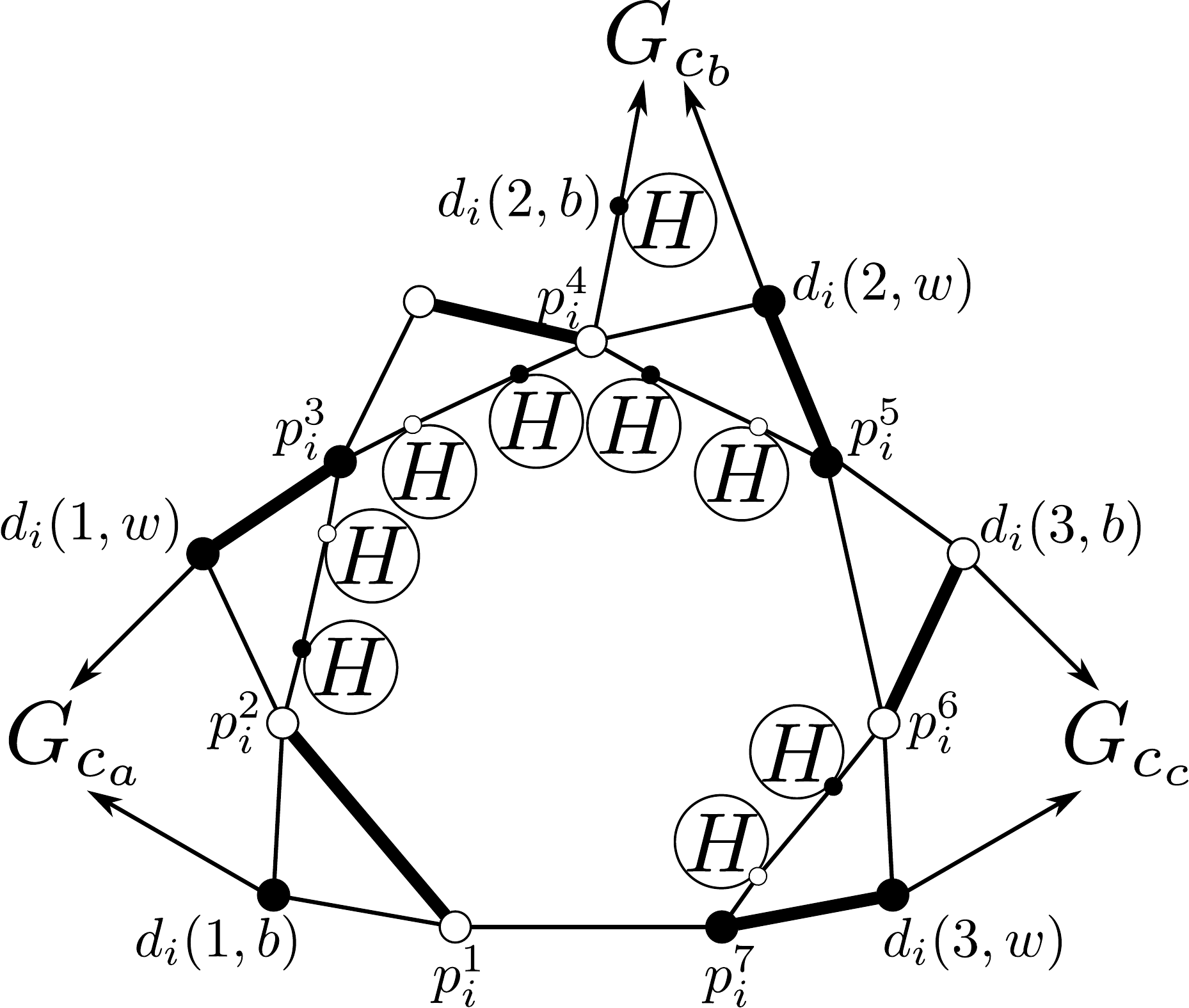}
			\caption{$p_i^1p_i^2 \in M$.}
			\label{fig:VarGadgetMathing1}
		\end{subfigure} 
		\begin{subfigure}[b]{.49\textwidth}
			\centering
			\includegraphics[width=.68\textwidth]{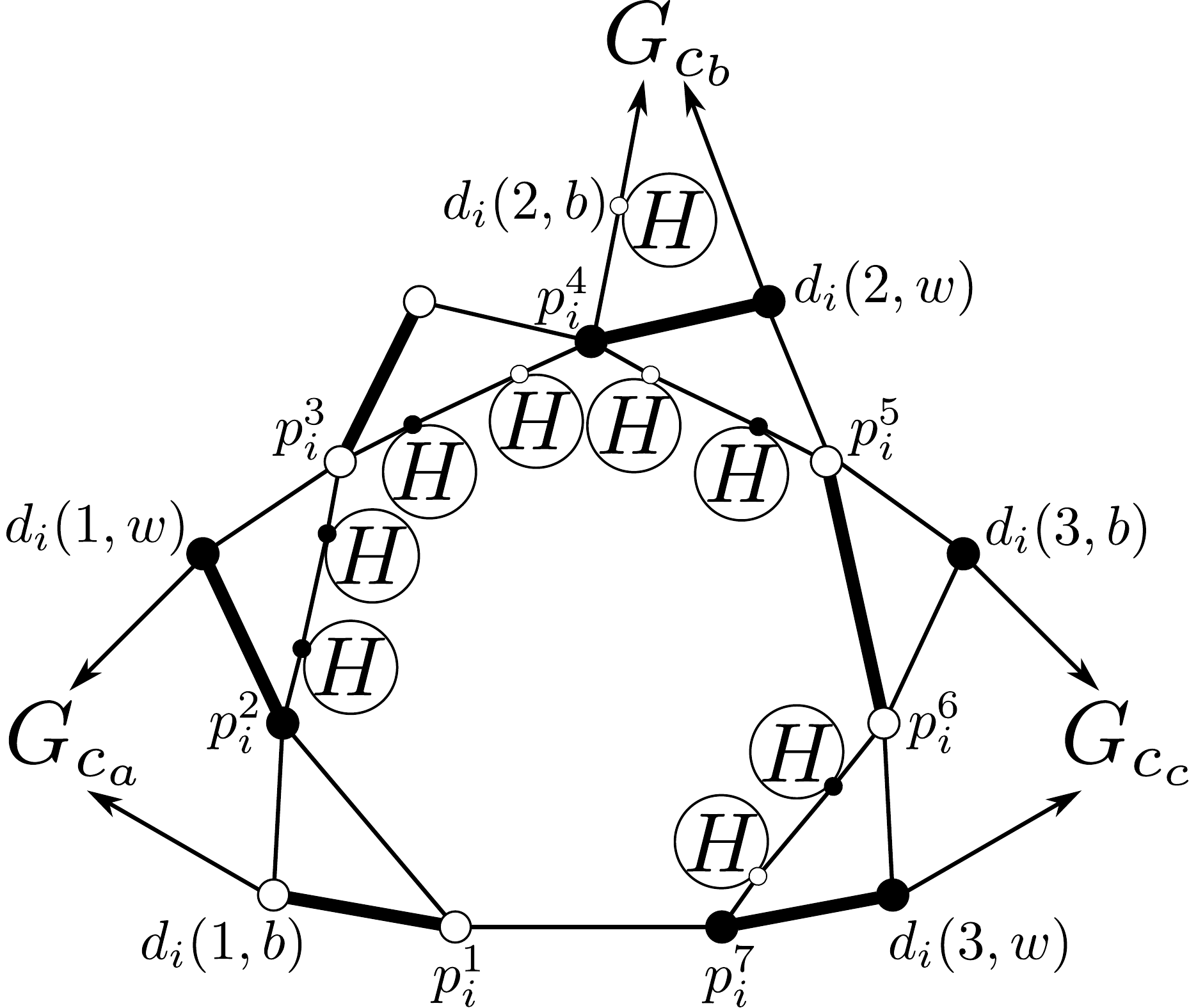}
			\caption{$p_i^5p_i^6 \in M$.}
			\label{fig:VarGadgetMathing2}
		\end{subfigure}
		\caption{All configurations by removing a bipartizing matching~$M$ of a variable gadget.
		}
		\label{fig:VarGadgetMathing}
	\end{figure}
	
	\noindent \textbf{Proof of Theorem~\ref{thm:np-planar4}.}
	Let~$G$ be the graph obtained by the construction in Theorem~\ref{thm:NPCompDeg5}.
	Since the only vertices of degree~5 are those~$p^4_i$ in the variable gadgets,
	we slightly modify the variable gadget as in Figure~\ref{fig:VarGadgetNP4}.
	In Figure~\ref{fig:Head} we can see that vertex~$h_6$ has degree~$3$, which allows us to use it to connect the variable gadget to the clause one.
	Figure~\ref{fig:VarGadgetNP4Match1} and Figure~\ref{fig:VarGadgetNP4Match2} show the possible bipartizing matchings of the new variable gadget.
	Since such configurations are analogous to those of the original variable gadget, with respect to the vertex that connect to clause gadgets, the theorem follows.
	
	\begin{figure}[h]
		\begin{subfigure}[b]{.31\textwidth}
			\centering
			\includegraphics[width=\textwidth]{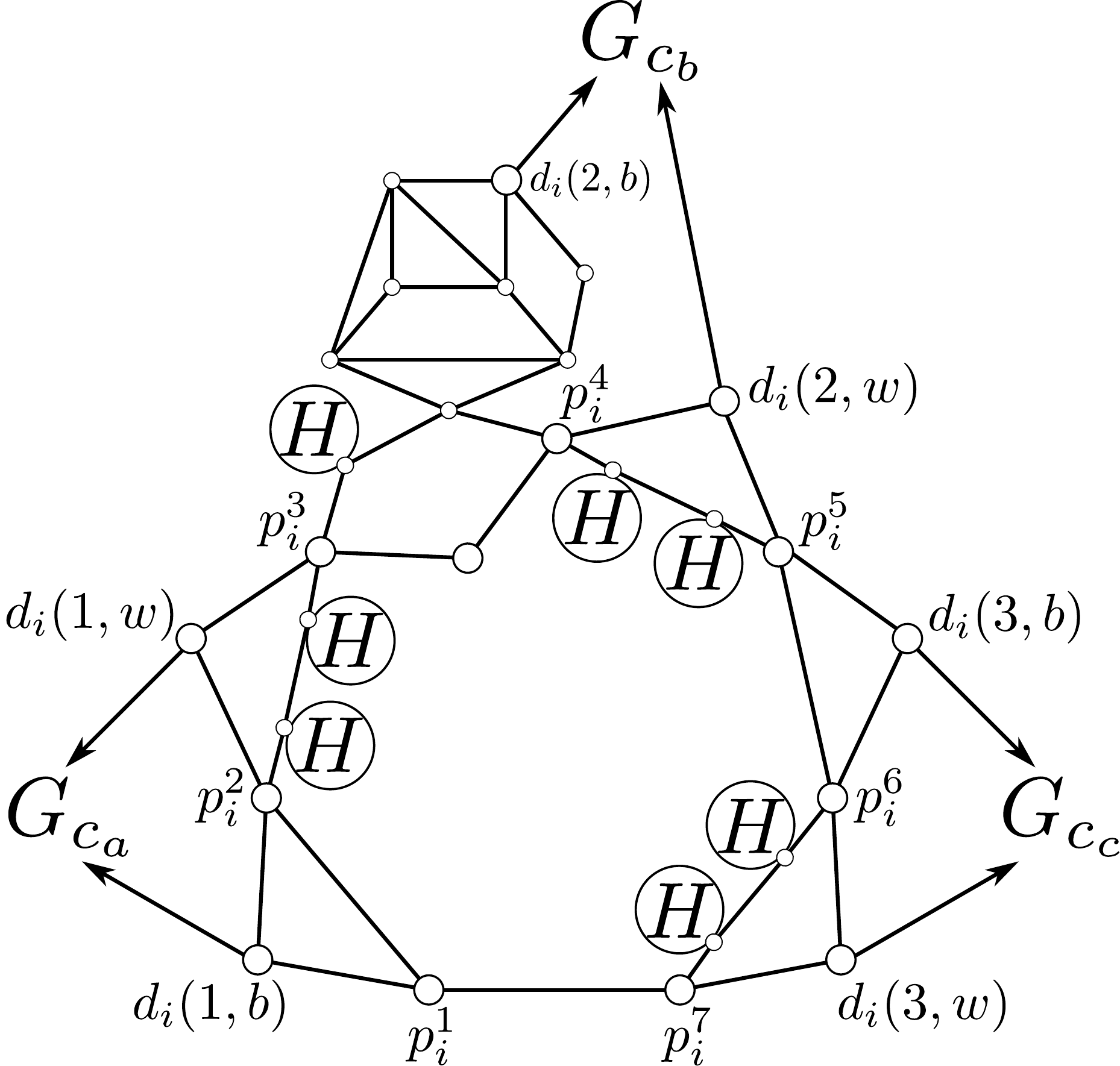}
			\caption{Modified variable gadget.}
			\label{fig:VarGadgetNP4}
		\end{subfigure}
		\begin{subfigure}[b]{.31\textwidth}
			\centering
			\includegraphics[width=\textwidth]{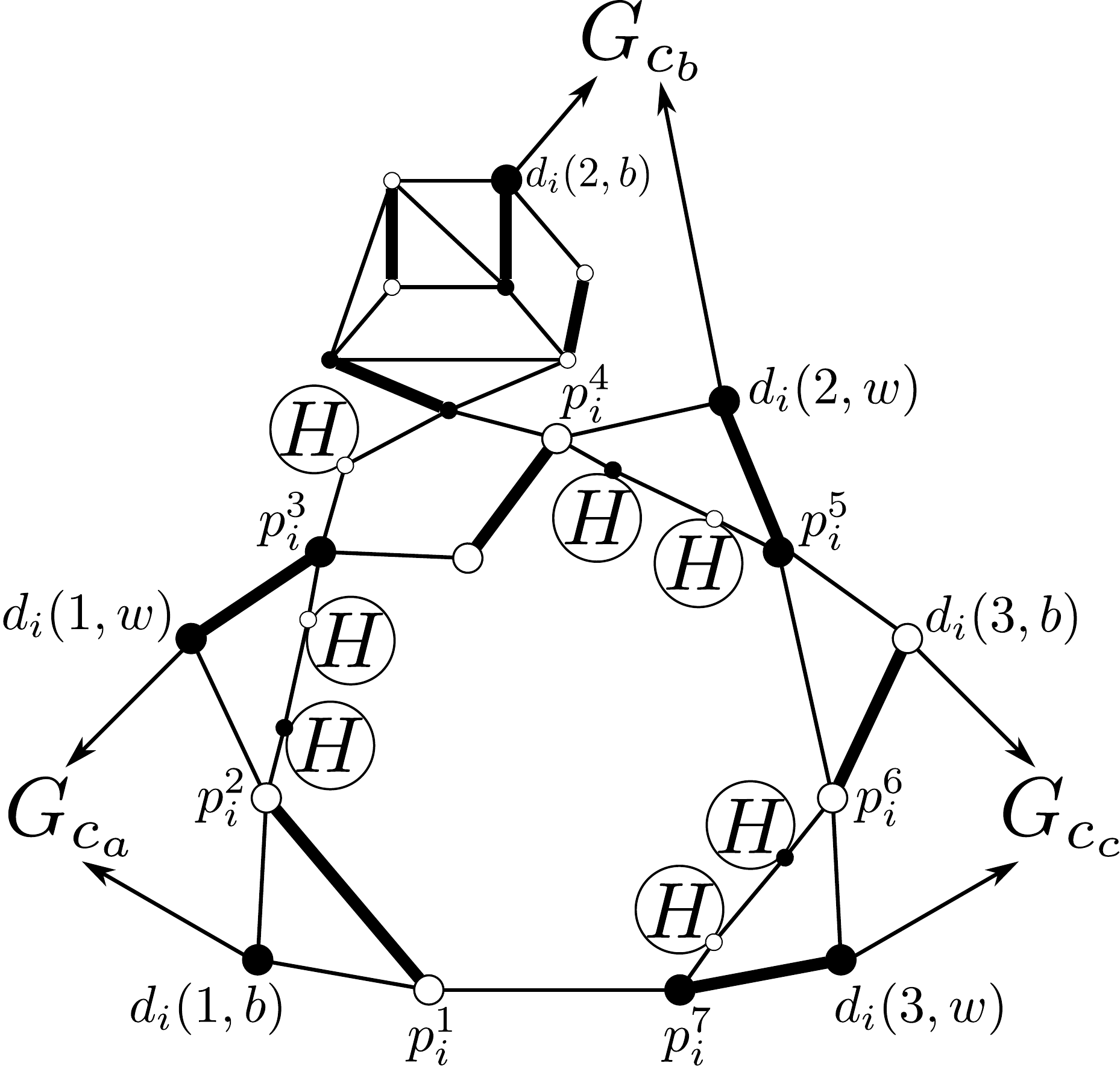}
			\caption{$p_i^1p_i^2 \in M$.}
			\label{fig:VarGadgetNP4Match1}
		\end{subfigure}
		\begin{subfigure}[b]{.31\textwidth}
			\centering
			\includegraphics[width=\textwidth]{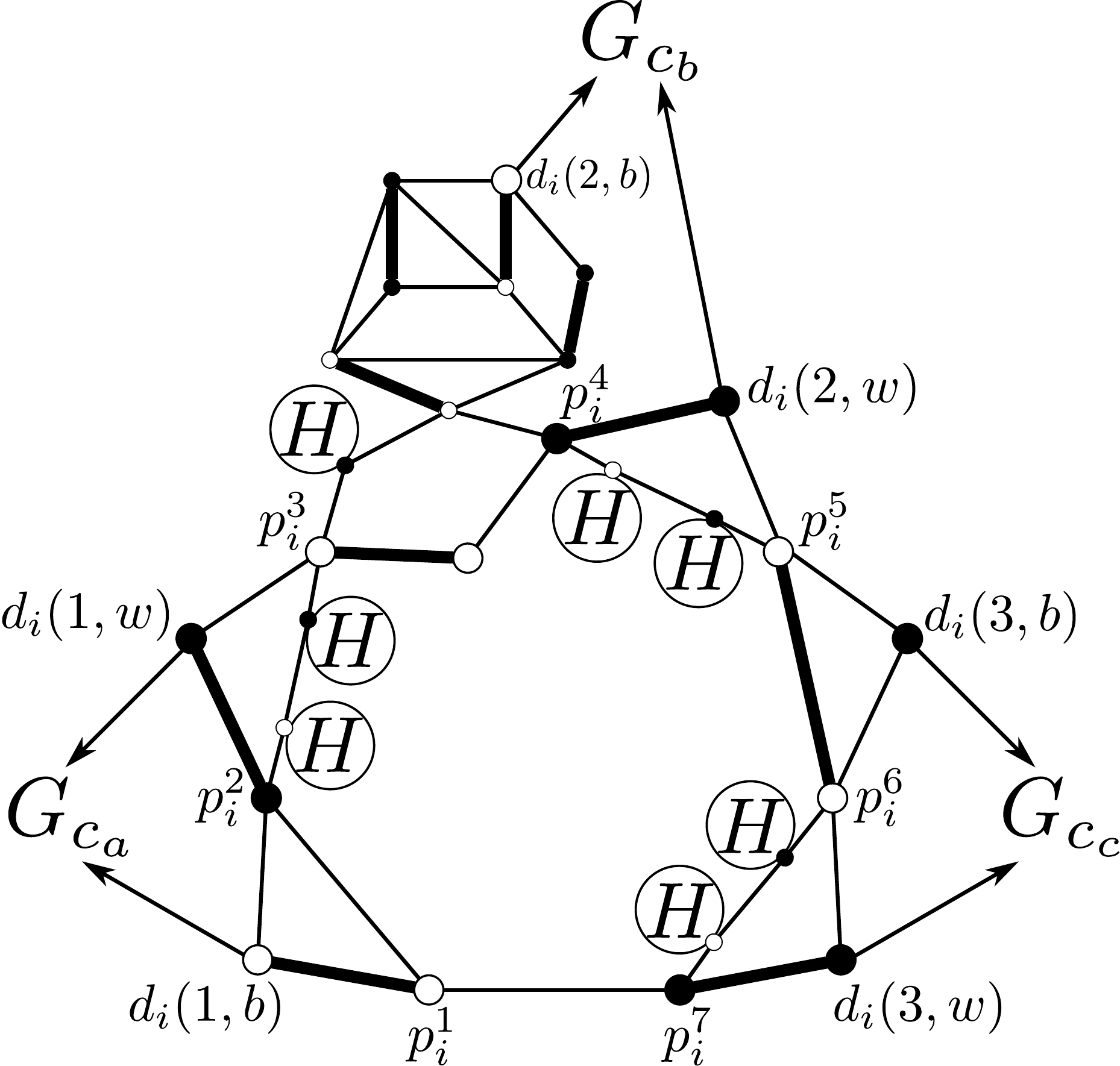}
			\caption{$p_i^5p_i^6 \in M$.}
			\label{fig:VarGadgetNP4Match2}
		\end{subfigure}
		\caption{The modified variable gadget and its all configurations given by removing a bipartizing matching~$M$.
		}
		\label{fig:MVarGadgetMathing}
	\end{figure}
	
	\section{Positive Results}
	\label{sec:PolyResults}

	\subsection{Graphs having Bounded Dominating Sets.}
	\label{subsec:BoundDomSet}
	
	A \textit{dominating set}~$S \subseteq V(G)$ of a graph~$G=(V,E)$ is a vertex set such that all vertices of~$V(G) \setminus S$ has a neighbor in~$S$.
	The cardinality of a minimum dominating set of~$G$ is the \textit{domination number of~$G$}.
	Now, consider that the domination number of the input graph~$G$ is bounded by a constant~$k$.
	
	\begin{proof}[Theorem~\ref{thm:polynomial}\textnormal{(a)}]
		A dominating set of order at most~$k$ can be found in time~$O(n^{k+2})$, by enumerating each vertex subset
		of size~$k$ and checking in time~$O(n(G)+m(G))$ whether it is a dominating set or not.
		Let~$D$ be such a dominating set of~$G$ of order at most~$k$.
		Let~$\mathcal{P_D}$ be the set of all bipartitions~$P_D$ of~$D$ into sets~$A_D$ and~$B_D$,
		such that~$G[A_D]$ and~$G[B_D]$ do not have any vertex of degree~2.
		Note that~$|\mathcal{P_D}| = O(2^k)$.
		
		Let~$P_D \in \mathcal{P_D}$ be a bipartition of~$D$.
		We partition all the other vertices of~$V(G)\setminus V(D)$ in such a way that~$P_D$ defines a bipartition
		of~$G-M_D$, if one exists, where~$M_D$ is a matching that will be removed, given the choice of~$D$.
		We do the following tests and operations for each vertex~$v \in V(G)\setminus V(D)$:
		\begin{itemize}
			\item If~$d_{G[A_D \cup \{v\}]}(v) \geq 2$ and~$d_{G[B_D]}(v) \geq 2$, then~$P_D$ is not a valid partition;
			\item If~$d_{G[A_D \cup \{v\}]}(v) \geq 2$, then~$B_D \gets B_D \cup \{v\}$;
			\item If~$d_{G[B_D \cup \{v\}]}(v) \geq 2$, then~$A_D \gets A_D \cup \{v\}$.
		\end{itemize}
		
		Iteratively we allocate the vertices of~$V(G) \setminus V(D)$ as described above into the
		respective sets~$A_D$ e~$B_D$, or we stop if it is not possible to acquire a valid bipartition.
		After these operations, set~$V' = V(G) \setminus \{A_D \cup B_D\}$ can be partitioned into three sets:
		\begin{itemize}
			\item $X = \{u \in V' :d_{G[A_D \cup \{u\}]}(u)=1 \mbox{ and } d_{G[B_D \cup \{u\}]}(u)=0\}$;
			\item $Y = \{u \in V' :d_{G[A_D \cup \{u\}]}(u)=0 \mbox{ and } d_{G[B_D \cup \{u\}]}(u)=1\}$;
			\item $Z = \{u \in V' :d_{G[A_D \cup \{u\}]}(u)=1 \mbox{ and } d_{G[B_D \cup \{u\}]}(u)=1\}$;
		\end{itemize}
		
		Since every vertex in~$V(G)\setminus V(D)$ has a neighbor in~$D$, it follows that the neighborhood of all the vertices
		of~$V'' = X\cup Y \cup Z$ in~$A_D \cup B_D$ is in~$D$.
		In this way, we can make a choice of a matching~$M_D$ to be removed, where the vertices of~$V''$
		are allocated either in~$A_D$ or in~$B_D$, and~$G-M_D$ is bipartite.
		Since each vertex of~$D$ can be matched to at most one vertex of~$V''$, there are~$O\left((n-k)^k\right)$ possibilities of choices for~$M_D$.
		
		Hence we obtain the following complexity:
		\[O\left(\sum_{i=1}^{k} n^{i+2} \cdot 2^i \cdot (n-i)^i\right) = O\left(k \cdot 2^k \cdot n^{k+2} \cdot (n-k)^k\right) = O\left(n^{2k+2}\right). \pushQED{\qed} \qedhere\]
	\end{proof}
	
	Theorem~\ref{thm:polynomial}(a) allows us to prove Theorem~\ref{thm:polynomial}(b).
	\begin{proof}[Theorem~\ref{thm:polynomial}\textnormal{(b)}]
		Every connected $P_5$-free graph has a dominating clique or a dominating $P_3$~\cite{Camby2016}, and graphs in~$\mathcal{BM}$ do not admit~$K_5$ as a subgraph.
		Thus, $P_5$-free graphs in~$\mathcal{BM}$ have domination number at most four.
	\end{proof}

	\subsection{Graphs with Only Triangles as Odd Cycles.}
	
	Consider now a slightly general version of \BM, where some edges are
	forbidden to be in any bipartizing matching.
	
	\begin{flushleft}
		\fbox{
			\begin{minipage}{.95\textwidth}
				\noindent {\sc Allowed Bipartizing Matching (ABM)}\\
				{\bf Instance:} A graph~$G$ and a set~$F$ of edges of~$G$. \\
				{\bf Task:} Decide whether~$G$ has a bipartizing matching~$M$ that does not intersect~$F$, and determine such a matching if it exists.
		\end{minipage}}
	\end{flushleft}
	
	A matching~$M$ as in {\sc ABM} is called an \textit{allowed bipartizing matching} of~$(G,F)$.
	
	We may clearly assume~$G$ as connected and bridge-free.
	Moreover, note that if~$(G,F)$ has an allowed bipartizing matching, then~$G \in \BM$.
	
	\begin{proof}[Theorem~\ref{thm:polynomial}\textnormal{(c)}]
		Let~$G$ be a graph having no $C_{2k+1}$, for~$k>1$, and let~$F \subseteq E(G)$.
		
		First, consider~$G$ a non-bipartite graph with no cut vertex, and let~$v_1v_2v_3v_1$ be an triangle of~$G$.
		Without loss of generality, we can assume that there is a vertex, say $v_1$, such that~$\{v_1v_2, v_1v_3\}$ is not an edge cut, otherwise $G$ would be a triangle.
		Then~$G-\{v_1v_2, v_1v_3\}$ has a path~$P$ from~$v_1$ to~$\{v_2, v_3\}$.
		Consider~$P$ as a longest one of length at least~$2$, and let~$v_2$ be the first vertex reached by~$P$ between~$v_2$ and~$v_3$.
		Thus~$P$ must be of the form~$v_1uv_2$, otherwise either~$G[V(P) \cup\{v_3\}]$ or~$G[V(P)]$ contains an odd cycle of length at least~$5$,
		when~$P$ has either an even or odd number of vertices, respectively.
		
		If~$w \in V(G)$ has exactly one neighbor~$z \in \{v_1,v_2,v_3,u\}$, then~$w$ is in a path~$P'$ of length at least
		two between~$z$ and~$z' \in \{v_1,v_2,v_3,u\}$, $z \neq z'$.
		Hence~$P' \cup \{v_1,v_2,v_3,u\}$ contains an odd cycle of length at least~5.
		Hence consider that~$w$ has at least two neighbors in~$\{v_1,v_2,v_3,u\}$.
		If~$uv_3 \in E(G)$, then~$G[\{v_1,v_2,v_3,u\}]$ is a~$K_4$, which implies that~$V(G) = \{v_1,v_2,v_3,u\}$, since~$\{w,v_1,v_2,v_3,u\}$
		induces an odd cycle.
		Thus~$G$ has an allowed bipartizing matching if and only if a maximal matching of~$G$ no intersecting~$F$.
		Otherwise if~$uv_3 \notin E(G)$, then~$w$ must be adjacent to either~$u$ and~$v_3$ or to~$v_1$ and~$v_2$, and no other vertex in~$\{v_1,v_2,v_3,u\}$.
		Moreover, the vertices adjacent to both~$u$ and~$v_3$ induce an independent set of~$G$, as well as the vertices	adjacent to both~$v_1$ and~$v_2$.
		In this case we can see that~$G$ has an allowed bipartizing matching if and only if $v_1v_2 \notin F$.
		
		Now, we consider a block decomposition of~$G$ with block-cut tree~$T$.
		Let~$B$ be a block containing exactly one cut-vertex~$v$, that is, $B$ is a leaf in~$T$.
		If~$(B, F)$ has an allowed bipartizing matching which is not incident to~$v$, then~$(G,F)$ has an allowed bipartizing matching if
		and only if~$(G^{'}, F)$ also admits one, where $G^{'} = \left(\left(V(G)\setminus V(B)\right) \cup\{v\}, E(G)\setminus E(B)\right)$. 
		Otherwise, $(G,F)$ has an allowed bipartizing matching if and only if~$(B, F)$ and $(G^{''},F^{''})$ admit allowed bipartizing matchings,
		where~$G^{''} = \left(\left(V(G)\setminus V(B)\right) \cup\{v\}, E(G)\setminus E(B)\right)$ and~$F^{''} = F \cup \{uv~|~u\in N_{G^{''}}(v)\}$.
		
		As in a block the desired matchings can be found, if any exists, in polynomial time, it is easy to see that we can solve {\sc ABM} in polynomial time.
	\end{proof}
	
	\section{Fixed-Parameter Tractability}
	\label{sec:FixedTract}
	
	In this section, we consider the parameterized complexity of \BM, and present an analysis of its complexity when parameterized by some classical parameters.	
	\begin{definition}
		The \textit{clique-width} of a graph~$G$, denoted by~$cwd(G)$, is defined as the minimum number of labels needed to construct~$G$, using the following four operations~\cite{Br05}:		
		\begin{enumerate}
			\item Create a single vertex~$v$ with an integer label~$\ell$ (denoted by~$\ell(v)$);
			\item Disjoint union of two graphs (i.e. co-join) (denoted by~$\oplus$);
			\item Join by an edge every vertex labeled~$i$ to every vertex labeled~$j$ for~$i \neq j$ (denoted by~$\eta(i,j)$);
			\item Relabeling all vertices with label~$i$ by label~$j$ (denoted by~$\rho(i,j)$).
		\end{enumerate}
	\end{definition}
	
	Courcelle et al.~\cite{Co00} stated that for any graph~$G$ with clique-width bounded by a constant~$k$, and for each graph property~$\Pi$ that can be formulated in a {\em  monadic second order logic}~($MSOL_1$), there is an~$O(f(cwd(G))\cdot n)$ algorithm that decides if $G$ satisfies $\Pi$~\cite{C90,C93,C97,Co00,Co11}.
	In~$MSOL_1$, the graph is described by a set of vertices~$V$ and a binary adjacency relation~$edge(.,.)$, and the graph property in question may be defined in terms of sets of vertices of the given graph, but not in terms of sets of edges.
	Therefore, in order to show the fixed-parameter tractability of \BM when parameterized by clique-width, it remains to show that the related property is $MSOL_1$-expressible.
	
	\begin{proof}[Theorem~\ref{thm:cwd}]
		Remind that the problem of determining whether~$G$ admits an odd decycling matching is equivalent to determine whether~$G$ admits an~$(2,1)$-coloring, which is a~$2$-coloring.
		Thus, it is enough to show that the property ``$G$ has an~$(2,1)$-coloring'' is $MSOL_1$-expressible.
		
		We construct a formula~$\varphi(G)$ such that~$G\in \BM \Leftrightarrow \varphi(G)$ as follows:
		\begin{align*}
		\exists~S_1, S_2 \subseteq V(G)~:
			& ~(S_1~\cap~S_2~=~\emptyset)~\wedge \\
			& ~(S_1~\cup~S_2~=~V(G))~\wedge \\
			& ~(\forall~v_1\in~S_1 [~\nexists~u_1,w_1\in S_1:~ (u_1\neq w_1) \wedge~edge(u_1,v_1) \wedge~edge(w_1,v_1)]) ~\wedge\\
			& ~(\forall~v_2\in~S_2 [~\nexists~u_2,w_2\in S_2 :~(u_2\neq w_2) \wedge~edge(u_2,v_2) \wedge~edge(w_2,v_2)]). \pushQED{\qed} \qedhere
		\end{align*}
	\end{proof}		
	Since clique-width generalizes several graph parameters~\cite{L12}, we obtain the following corollary.
	\begin{corollary}\label{cor:parameterized}
		{\sc Odd Decycling Matching} is in {\sf FPT} when parameterized by the following parameters:
		\begin{itemize}
			\item neighborhood diversity;
			\item treewidth;
			\item pathwidth;
			\item feedback vertex set;
			\item vertex cover.
		\end{itemize}
	\end{corollary}
	
	Since~$K_5$ is a forbidden subgraph, chordal graphs in~$\mathcal{BM}$ have bounded treewidth~\cite{rs86}, and thus, \BM is polynomial-time solvable for this class and Corollary~\ref{cor:chordal} follows.
	
	Courcelle's theorem is a good classification tool, however it does not provide a precise {\sf FPT}-running time.
	The next result shows the exact upper bound for \BM parameterized by the \textit{vertex cover number}~$vc(G)$ of~$G$.
	A \textit{vertex cover}~$S \subseteq V(G)$ of~$G$ is a vertex subset such that~$G[V(G)\setminus S]$ has no edge, that is, it is an independent set.
	The vertex cover number is the cardinality of a minimum vertex cover.
	\begin{theorem}\label{thm:VertexCover}
		\BM admits a~$2^{O(vc(G))}\cdot n(G)$ algorithm.
	\end{theorem}
	\begin{proof}
		Let~$S$ be a vertex cover of~$G$ such that~$|S| = vc(G)$.
		The algorithm follows in a similar way to the algorithm in the proof of Theorem~\ref{thm:polynomial}\textnormal{(a)}.
		Let~$\mathcal{P_S}$ be the set of all bipartitions~$P_S$ of~$S$ into sets~$A_S$ and~$B_S$, such that~$S[A_S]$ and~$S[B_S]$ do not have any vertex of degree~2.
		Note that~$|\mathcal{P_S}| = O(2^k)$.
		
		For each~$P_S \in \mathcal{P_S}$, we will check if an odd decycling matching of~$G$ can be obtained from~$P_S$ by applying the following operations:
		
		For each vertex~$v \in V(G)\setminus V(S)$ do
		\begin{itemize}
			\item If~$d_{A_S}(v) \geq 2$ and~$d_{B_S}(v) \geq 2$, then~$P_S$ is not a valid partial partition;
			\item If~$d_{A_S}(v) \geq 2$, then~$A_S \gets A_S \cup \{v\}$;
			\item If~$d_{B_S}(v) \geq 2$, then~$B_S \gets B_S \cup \{v\}$.
		\end{itemize}		
		After that, if for all vertices the first condition is not true, then $V(G) \setminus V(A_S \cup B_S)$ can be partitioned into three sets:
		\begin{itemize}
			\item $X = \{u \in V(G) \setminus V(A_S \cup B_S):d_{A_S}(u)=1 \mbox{ and } d_{B_S}(u)=0\}$;
			\item $Y = \{u \in V(G) \setminus V(A_S \cup B_S):d_{A_S}(u)=0 \mbox{ and } d_{B_S}(u)=1\}$;
			\item $Z = \{u \in V(G) \setminus V(A_S \cup B_S):d_{A_S}(u)=1 \mbox{ and } d_{B_S}(u)=1\}$;
		\end{itemize}		
		Since~$V(G)\setminus V(S)$ is an independent set, it follows that all edges of vertices in~$X \cup Y$ can remain in the graph~$G$.
		For each~$z \in Z$, denote by $a_z\in A_S$ and~$b_z \in B_S$ the neighbors of~$z$ in~$G$.
		
		Now, we apply a bounded search tree algorithm.
		While~$G[A_S]$ and~$G[B_S]$ have both maximum degree equal to one, and~$Z\neq \emptyset$ do.
		Remove a vertex~$z\in Z$ and apply recursively the algorithm for the following cases:
		\begin{enumerate}
			\item $z$ is added to $A_S$, and all vertices in $Z\cap N(a_z)$ is added to $B_S$.
			\item $z$ is added to $B_S$, and all vertices in $Z\cap N(b_z)$ is added to $A_S$.
		\end{enumerate}
		Note that the search tree~$T$ has height equals~$vc(G)+1$.
		Finally, if~$T$ has a leaf representing a configuration with~$G[A_S]$ and~$G[B_S]$ having both maximum degree equal to one, and~$Z = \emptyset$, then~$G \in \mathcal{BM}$.
	\end{proof}
	
	Now we analyze the parameterized complexity of \BM considering the \textit{neighborhood diversity number}, $nd(G)$, as parameter.	
	\begin{definition}
		A graph~$G(V,E)$ has neighborhood diversity~$nd(G)=t$ if we can partition~$V$ into~$t$ sets~$V_1, \dotsc, V_t$ such that, for every~$v\in V$ and all~$i\in {1, \dotsc, t}$, either~$v$ is adjacent to every vertex in~$V_i$ or it is adjacent to none of them.
		Note that each part~$V_i$ of~$G$ is either a clique or an independent set.
	\end{definition}	
	The neighborhood diversity parameter is a natural generalization of the vertex cover number.
	In 2012, Lampis~\cite{L12} showed that for every graph~$G$ we have~$nd(G)\leq 2^{vc(G)}+vc(G)$.
	The optimal neighborhood diversity decomposition of a graph~$G$ can be computed in~$\mathcal{O}(n^3)$ time~\cite{L12}.
	
	\begin{theorem}\label{thm:NeighborDivers}
		\BM admits a kernel with at most~$2\cdot nd(G)$ vertices when parameterized by neighborhood diversity number.
	\end{theorem}
	\begin{proof}
		Given an instance~$(G,F)$ of {\sc ABM} such that~$G$ is a graph and~$F\subseteq E(G)$ a set of forbidden edges.
		The kernelization algorithm consists on applying the following reduction rules:
		\begin{enumerate}
			\item \textit{If~$G$ contains a~$K_5$, then~$G$ has no allowed odd decycling matching; otherwise}
			\item \textit{If a part~$V_i$ induces a~$K_3$ and exist two vertices in~$V(G)\setminus V_i$ adjacent to~$V_i$, then~$G$ has no allowed bipartizing matching; otherwise}
			\item \textit{If a subgraph of $G$ induces either a~$K_3$ or a~$K_4$ and does not admit an allowed bipartizing matching, then~$G$ has no allowed bipartizing matching; otherwise}
			\item \textit{Remove all parts isomorphic to a~$K_4$;}
			\item \textit{Remove all isolated parts isomorphic to a~$K_3$;}
			\item \textit{If~$V_i$ is a part that induces a~$K_3$ and $v \in V(G)\setminus V_i$ is adjacent to~$V_i$ (note that $\{v\}$ is a part), then remove~$V_i$ and~$F \gets F \cup \{uv: u \in N_G(v) \setminus V_i\}$;}
			\item \textit{If a part~$V_i$ induces an independent set of size at least~3, then contract it into a single vertex~$v_i$ (without parallel edges) and forbids all of its incident edges.}
		\end{enumerate}
		It is easy to see that all reduction rules can be applied in polynomial time, and after applying them any remaining part has size at most two.
		As the resulting graph~$G'$ has~$nd(G')\leq nd(G)$ then~$|V(G')|\leq 2\cdot nd(G)$.
		Thus, it remains to prove that the application of each reduction rule is correct.
		As~$K_5$ and~$K_5 - e$ are forbidden subgraphs, and any bipartizing matching of a~$K_4$ is a perfect matching, then rules~$1, 2, 3, 4, 5$, and~$6$ can be applied in this order.
		Finally, the correctness of rule~$7$ follows from the following facts: $(i)$ if~$G'$ has an allowed bipartizing matching, then~$G$ has also an allowed bipartizing matching, because bipartite graph class is closed under the operation of replacing vertices by a set of false twins, which have the same neighborhood as the replaced vertex; $(ii)$ if~$G'$ does not admit an allowed bipartizing matching then~$G$ also does not admit an allowed bipartizing matching, because if a contracted single vertex~$v_i$ is in an odd cycle in~$G'$, then even replacing~$v_i$ by~$V_i$ ($|V_i|\geq 3$) and removing some incident edges of~$V_i$ that form a matching, some vertex of~$V_i$ remains in an odd cycle.
	\end{proof}
	
	\section{On $(2,d)$-Coloring Planar Graphs of Bounded Degree}
	\label{sec:2,d-coloring}
	
	\begin{figure}[t]
		\centering
		\includegraphics[width=0.6\textwidth]{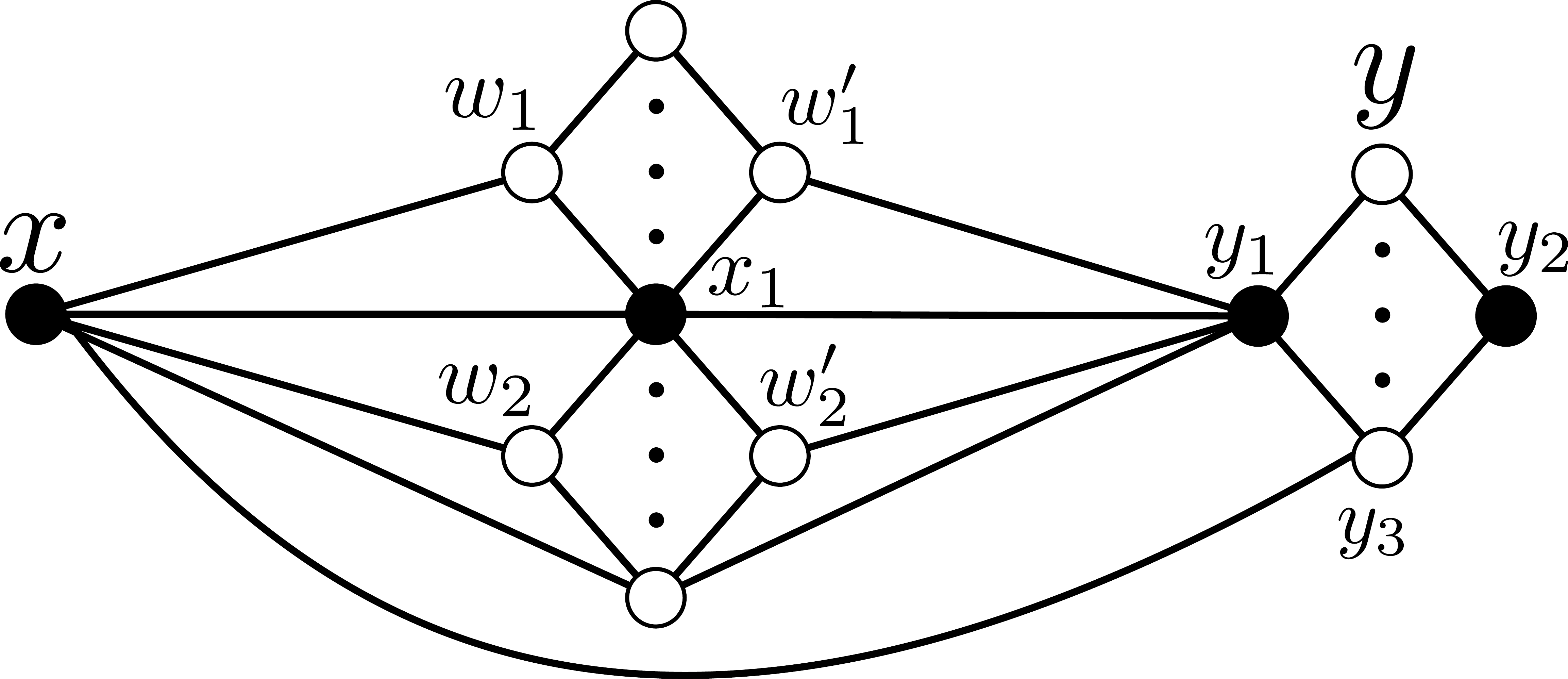}
		\caption{The gadget used in the reduction of Theorem~\ref{thm:2,2-coloring} $(2,2)$-colored with colors black and white.
			The vertices~$x$ and~$y$ must receive different colors in any $(2,2)$-coloring of~$H$.}
		\label{fig:2,2-coloring}
	\end{figure}
	
	In this section we present the proofs of Theorems~\ref{thm:2,2-coloring} and Corollary~\ref{cor:2,3-coloring}.
	The main idea in the proof of Theorem~\ref{thm:2,2-coloring} is to find a planar gadget~$H$ (see Figure~\ref{fig:2,2-coloring}) of maximum degree~6 containing two vertices~$x$ and~$y$	of degree at most~5 and such that they must receive different colors in any~$(2,2)$-coloring of~$H$.
	We prove this on Lemma~\ref{lem:2,2-coloring}.
	
	We complete the proof by a reduction from~$(2,1)$-coloring for a planar graph~$G$ of maximum degree~$4$, that is {\sf NP}-complete by Theorem~\ref{thm:np-planar4}.
	Let~$G'$ be the graph obtained by adding for each vertex~$v$ of~$G$ a gadget~$H_v$ as in Figure~\ref{fig:2,2-coloring}, such that~$v$ is adjacent to only~$x$ and~$y$ in~$H_v$.
	Therefore it follows that~$G'$ is a planar graph of maximum degree~6.
	
	If~$G \in \mathcal{BM}$, then each vertex~$v$ of~$G$ can be colored in such a way that~$v$ has at most one neighbor colored as itself and then~$v$ can have one more neighbor in~$H_v$ with the same color.
	We can complete the~$(2,2)$-coloring of~$G'$ as that one depicted in Figure~\ref{fig:2,2-coloring}, where the black and white vertices define the bipartition.
	Conversely, if~$G'$ admits a~$(2,2)$-coloring~$c$, then each vertex~$v \in V(G)$ has exactly one neighbor in~$H_v$ colored as itself, which implies that~$v$ admits at most one more neighbor in~$G$ colored as the same way.
	Hence~$c$ restricted to the vertices of~$G$ is a~$(2,1)$-coloring of~$G$, and the theorem follows.
	
	Next we present the proof of Lemma~\ref{lem:2,2-coloring}.
	In a first observation, for every pair of vertices~$u, v \in V(G)$ that share at least~$2d+1$ neighbors, it follows that both~$u$ and~$v$ must have the same color in any~$(2,d)$-coloring of~$G$.
	Otherwise, one of them has at least~$d+1$ neighbors with the same color as itself.
	Therefore, by simplicity, we omit the complete neighborhood in Figure~\ref{fig:2,2-coloring} of the vertices~$w_i$, $w_i'$, $y$, and~$y_3$, $i \in \{1,2\}$, where~$w_i$ ($y$) shares~$5$ neighbors with~$w_i'$ ($y_3$).
	We can see that~$w_i$, $w_i'$, $y_1$, $y_3$ and~$x_1$ have degree~$6$, while~$x$ and~$y$ have degree~5.
	
	\begin{lemma}\label{lem:2,2-coloring}
		In any~$(2,2)$-coloring of~$H$, $x$ and~$y$ must have distinct colors.
	\end{lemma}
	\begin{proof}
		As~$w_i$ and~$w_i'$ shares~$5$ neighbors, $i \in \{1,2\}$, it follows that they must receive the same color in any~$(2,2)$-coloring of~$H$.
		The same occurs for~$y$ and~$y_3$.
		In this way, $w_i$ and~$w_i'$ can be seen as the same vertex in terms of any~$(2,2)$-coloring of~$G$, implying that~$x$ and~$y_1$ also ``shares~$5$ neighbors'' with respect to the coloring.
		This implies that~$x$ and~$y_1$ must receive the same color as well.
		Now it is enough to see that at least one vertex in~$N_H(y_1) \setminus \{y, y_3\}$ must receive the same color as~$y_1$, otherwise~$x_1$ would have four vertices with the same color.
		Therefore, the color of~$y$ and~$y_3$ have to be different from that of~$y_1$, implying that~$x$ and~$y$ receive different colors.
	\end{proof}
	
	We can use the same approach in the proof of Theorem~\ref{thm:2,2-coloring} to obtain the bound of~$2d+2$ on the maximum degree on the {\sf NP}-completeness of $(2,3)$-color planar graphs.
	We just modify the gadget used in Theorem~\ref{thm:2,2-coloring} as in Figure~\ref{fig:2,3-coloring}.
	With the same arguments, we can see that~$y_1$ must be adjacent to at least two vertices of the same color in~$N_H(y_1)\setminus \{y, y_3\}$, otherwise~$x_1$ or~$x_2$ would have four neighbors of the same color as themselves.
	Since~$x$ and~$y$ have degree~7 and receive different colors in any~$(2,3)$-coloring of~$H$, we can use the same strategy of adding a gadget~$H$ of Figure~\ref{fig:2,3-coloring} to each vertex of a planar graph~$G$ of maximum degree~6, obtaining a new graph~$G'$ of maximum degree~8.
	We conclude the reduction from Theorem~\ref{thm:2,2-coloring}, that ensures the {\sf NP}-completeness for~$G$.
	In this way, it follows that Corollary~\ref{cor:2,3-coloring} follows.
	
	\begin{figure}[t]
		\centering
		\includegraphics[width=0.6\textwidth]{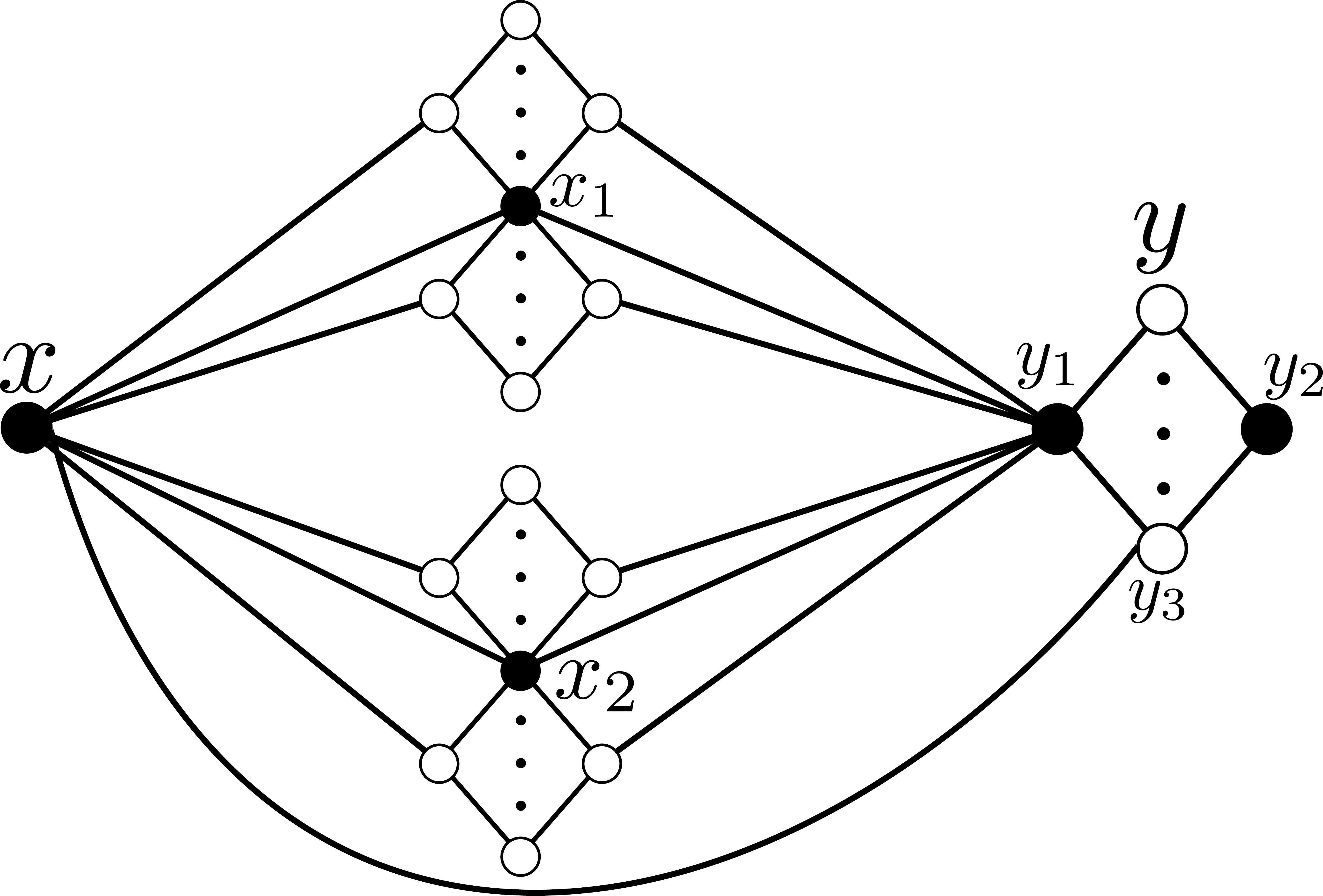}
		\caption{The gadget used in the reduction of Corollary~\ref{cor:2,3-coloring} $(2,3)$-colored with colors black and white.
			The vertices~$x$ and~$y$ must receive different colors in any $(2,3)$-coloring of~$H$.}
		\label{fig:2,3-coloring}
	\end{figure}
	
	It is not hard to see that we can extend the idea to values of~$d$ greater than~$3$.
	The central idea is in taking a vertex that has at least~$d-1$ vertices of the same color as itself in any~$(2,d)$-coloring.
	This holds for~$y_1$ in the both previous reduction.
	Then, we can add~$d-1$ structures in the neighborhood of~$y_3$, like in Figure~\ref{fig:2,3-coloring}, in order to ensure that~$y_3$ has~$d-1$ one neighbors colored as itself in each structure.
	We can obtain such neighbors by extending the neighborhood of each~$x_i$ as in Figure~\ref{fig:2,4-coloring} and Figure~\ref{fig:2,5-coloring}, for~$(2,4)$-coloring and~$(2,5)$-coloring, respectively.
	We can see that each vertex~$x_i$ cannot receive the same color as all of its neighbors less~$x$.
	So~$x$ has at least one neighbor colored as itself for each~$x_i$, $1 \leq i \leq d-1$ implying that~$x$ and~$y$ must receive different colors in any~$(2,d)$-color of~$H$.
	Hence we can use this gadget in order to obtain a~$(2,d)$-coloring of a planar graph~$G$ by adding~$H$ to each vertex~$v$ of~$G$, as in the proofs of Theorem~\ref{thm:2,2-coloring} and Corollary~\ref{cor:2,3-coloring}, where~$v$ is adjacent only to~$x$ and~$y$ in~$H$, by using the gadget for~$(2,d-1)$-color~$G$.
	
	It is not hard to see that the above construction cannot be applied in general, since the degree of~$x$ increases in a quadratic way on~$d$.
	Even so, we can see that it gives an upper bound on the maximum degree better than~$4d+3$ for~$d\leq 7$.
	
	\begin{figure}
		\begin{subfigure}[b]{.45\textwidth}
			\centering
			\includegraphics[width=\textwidth]{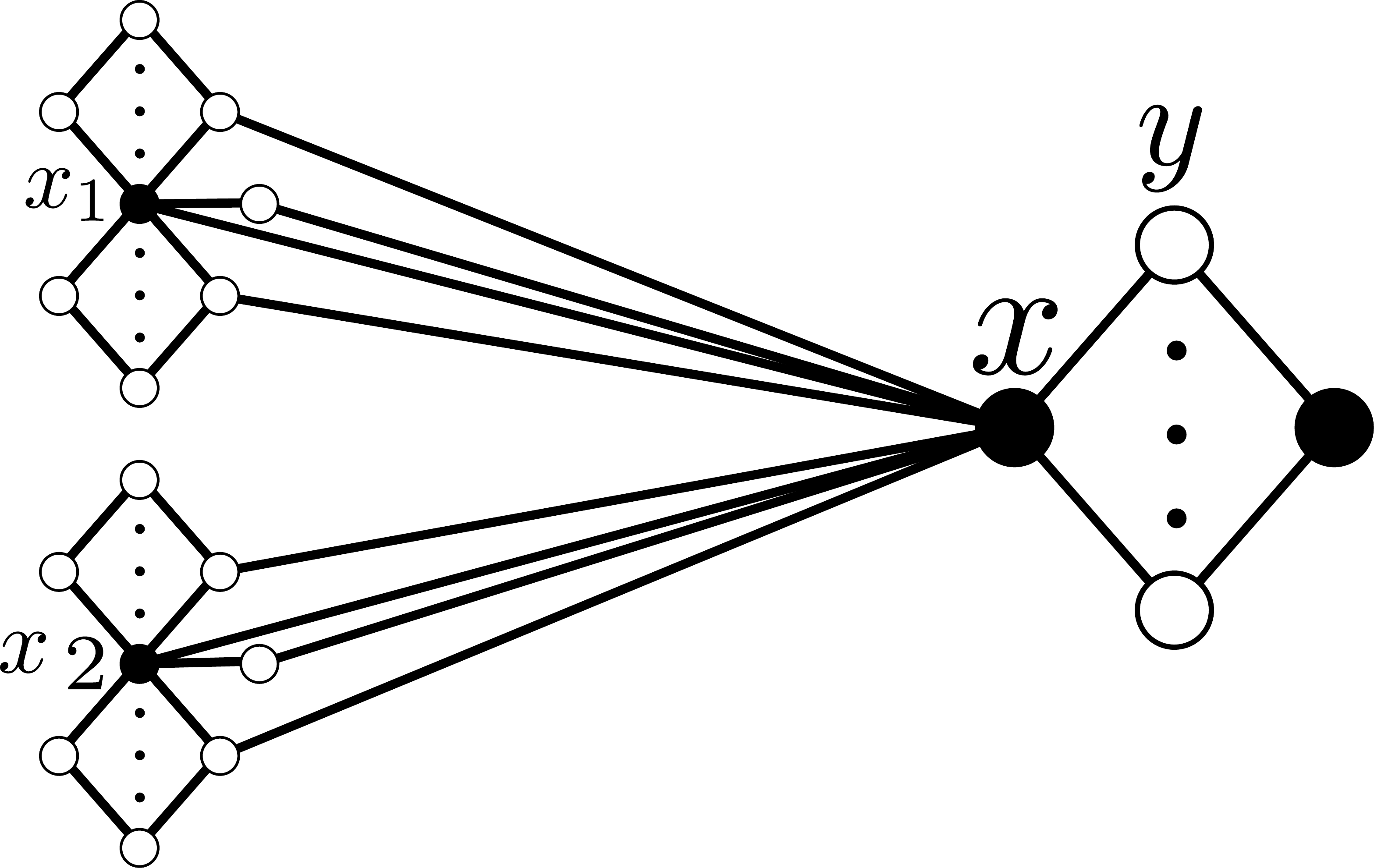}
			\caption{Gadget for $(2,4)$-coloring.}
			\label{fig:2,4-coloring}
		\end{subfigure} \qquad
		\begin{subfigure}[b]{.45\textwidth}
			\centering
			\includegraphics[width=\textwidth]{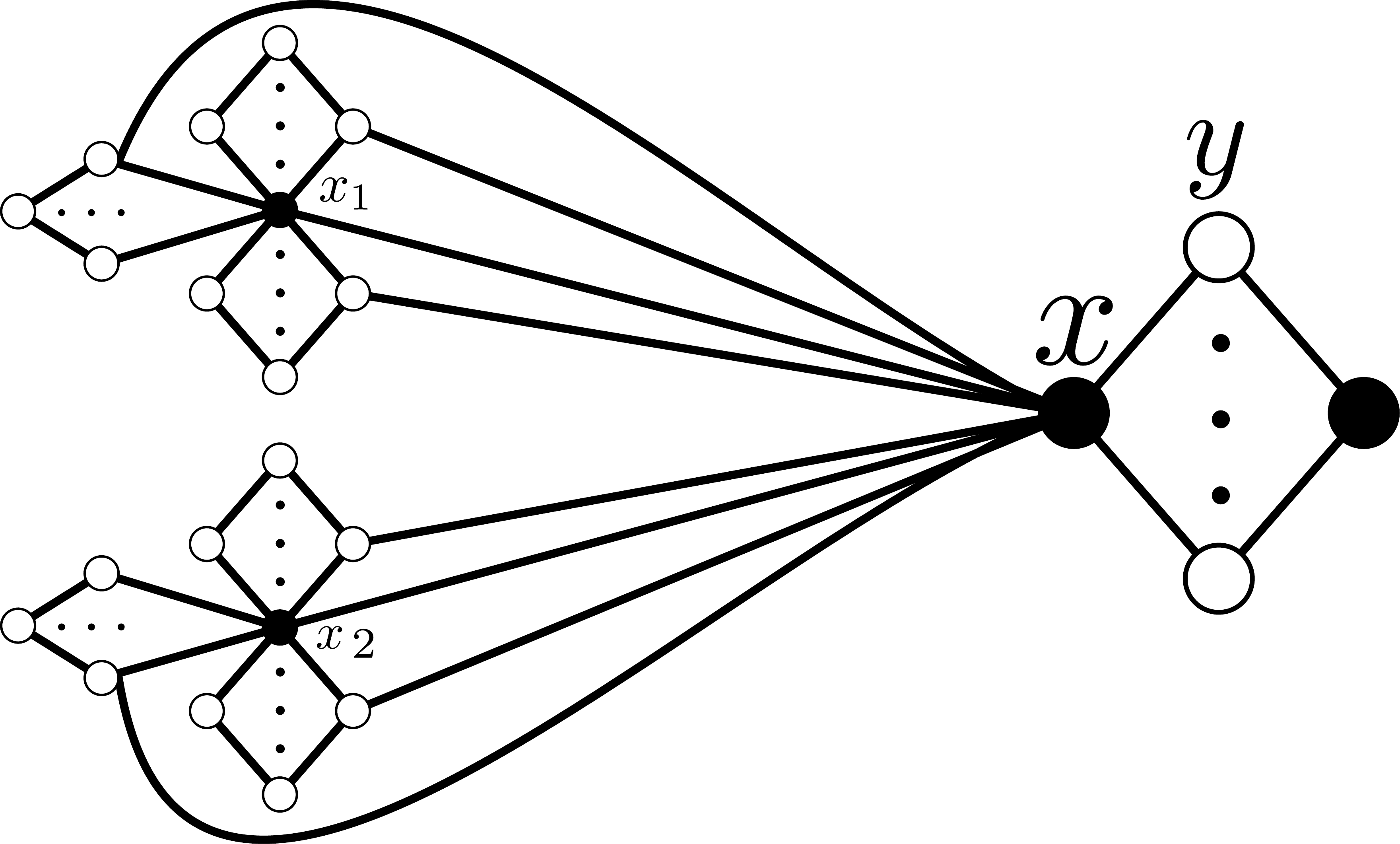}
			\caption{Gadget for $(2,5)$-coloring.}
			\label{fig:2,5-coloring}
		\end{subfigure}
		\caption{Generalized gadgets for~$(2,d)$-color planar graphs.
		}
		\label{fig:2,d-coloring}
	\end{figure}
	
	\section{Conclusion}
	\label{sec:conclusion}
	
	In this paper we have obtained the complete dichotomy on the computational complexity for~$(2,1)$-color a planar graph in terms of the maximum degree, where all graphs of maximum degree~$3$ admits a bipartizing matching and we prove the {\sf NP}-completeness for maximum degree~$4$.
	We extended the result for~$(2,2)$-coloring and~$(2,3)$-coloring.
	We left open Question~\ref{quest:2,d-coloring} restricted~$(2,d)$color planar graphs.
	We also give some parameterized complexity results, proving that bipartizing a graph by the removal of a matching is {\sf FPT} when parameterized by the clique-width, which results in polynomial-time algorithms for a number of important graph classes.
	
	Interesting properties regarding the chromatic number of graphs in~$\mathcal{BM}$ can be proposed.
	For example, which graphs~$G \in \mathcal{BM}$ are such that~$\chi(G-M) \leq \chi(G)$, for a bipartizing matching~$M$ of~$G$?
	Moreover, what is the maximum size of the gap between~$\chi(G-M)$ and~$\chi(G)$?


\begin{thebibliography}{10}
		
		\bibitem{Abdullah}
		{\sc Abdullah, A.}
		\newblock On graph bipartization.
		\newblock In {\em ISCAS '92\/} (1992), vol.~4, pp.~1847--1850.
		
		\bibitem{Agrawal}
		{\sc Agrawal, A., Jain, P., Kanesh, L., Misra, P., and Saurabh, S.}
		\newblock {Exploring the Kernelization Borders for Hitting Cycles}.
		\newblock In {\em 13th International Symposium on Parameterized and Exact
			Computation (IPEC 2018)\/} (Dagstuhl, Germany, 2019), vol.~115 of {\em
			Leibniz International Proceedings in Informatics (LIPIcs)}, Schloss
		Dagstuhl--Leibniz-Zentrum fuer Informatik, pp.~14:1--14:14.
		
		\bibitem{as}
		{\sc Alon, N., and Stav, U.}
		\newblock Hardness of edge-modification problems.
		\newblock {\em Theory Comput. Sci. 410}, 47-49 (2009), 4920--4927.
		
		\bibitem{Andrews85}
		{\sc Andrews, J., and Jacobson, M.}
		\newblock On a generalization of chromatic number.
		\newblock In {\em Proc. Sixteenth Southeastern International Conference on
			Combinatorics, Graph Theory and Computing\/} (1985), vol.~47, pp.~18--33.
		
		\bibitem{Angelini17}
		{\sc {Angelini}, P., {Bekos}, M.~A., {De Luca}, F., {Didimo}, W., {Kaufmann},
			M., {Kobourov}, S., {Montecchiani}, F., {Raftopoulou}, C.~N., {Roselli}, V.,
			and {Symvonis}, A.}
		\newblock Vertex-coloring with defects.
		\newblock {\em J. Graph Algor. Appl. 21}, 3 (2017), 313--340.
		
		\bibitem{Axenovich}
		{\sc Axenovich, M., Ueckerdt, T., and Weiner, P.}
		\newblock Splitting planar graphs of girth 6 into two linear forests with short
		paths.
		\newblock {\em J. Graph Theory 85}, 3 (2017), 601--618.
		
		\bibitem{Bang19}
		{\sc Bang-Jensen, J., and Bessy, S.}
		\newblock Degree-constrained 2-partitions of graphs.
		\newblock {\em Theo. Comp. Sci.\/} (2019).
		\newblock in press.
		
		\bibitem{B98}
		{\sc Bodlaender, H.~L.}
		\newblock A partial $k$-arboretum of graphs with bounded treewidth.
		\newblock {\em Theo. Comp. Sci. 209}, 1--2 (1998), 1--45.
		
		\bibitem{Bonamy2018}
		{\sc Bonamy, M., Dabrowski, K.~K., Feghali, C., Johnson, M., and Paulusma, D.}
		\newblock Independent feedback vertex set for \textsc{P}$_5$-free graphs.
		\newblock {\em Algorithmica\/} (2018).
		
		\bibitem{bky}
		{\sc Borodin, O., Kostochka, A., and Yancey, M.}
		\newblock On $1$-improper $2$-coloring of sparse graphs.
		\newblock {\em Discrete Math. 313}, 22 (2013), 2638--2649.
		
		\bibitem{Br05}
		{\sc Brandst{\"a}dt, A., Dragan, F.~F., Le, H.-O., and Mosca, R.}
		\newblock New graph classes of bounded clique-width.
		\newblock {\em Theo. Comput. Sys. 38}, 5 (2005), 623--645.
		
		\bibitem{bell06}
		{\sc Brandst{\"a}dt, A., Engelfriet, J., Le, H.-O., and Lozin, V.~V.}
		\newblock Clique-width for $4$-vertex forbidden subgraphs.
		\newblock {\em Theory Comput. Syst. 39}, 4 (2006), 561--590.
		
		\bibitem{bkm06}
		{\sc Brandst{\"a}dt, A., Klembt, T., and Mahfud, S.}
		\newblock {$P_6$- and triangle-free graphs revisited: structure and bounded
			clique-width}.
		\newblock {\em {Discrete Math. Theor. Comput. Sci.} 8\/} (2006), 173--188.
		
		\bibitem{bbd}
		{\sc Burzyn, P., Bonomo, F., and Dur\'{a}n, G.}
		\newblock \texttt{NP}-completeness results for edge modification problems.
		\newblock {\em Discrete Appl. Math. 154}, 13 (2006), 1824--1844.
		
		\bibitem{Camby2016}
		{\sc Camby, E., and Schaudt, O.}
		\newblock A new characterization of {$P_k$}-free graphs.
		\newblock {\em Algorithmica 75}, 1 (2016), 205--217.
		
		\bibitem{cnr}
		{\sc Choi, H.-A., Nakajima, K., and Rim, C.~S.}
		\newblock Graph bipartization and via minimization.
		\newblock {\em SIAM J. Discrete Math. 2}, 1 (1989), 38--47.
		
		\bibitem{Chuangpishit}
		{\sc Chuangpishit, H., Lafond, M., and Narayanan, L.}
		\newblock Editing graphs to satisfy diversity requirements.
		\newblock In {\em Combinatorial Optimization and Applications (COCOA 2018)\/}
		(Cham, 2018), Springer International Publishing, pp.~154--168.
		
		\bibitem{C90}
		{\sc Courcelle, B.}
		\newblock The monadic second-order logic of graphs. i. recognizable sets of
		finite graphs.
		\newblock {\em Inf. Comput. 85}, 1 (1990), 12--75.
		
		\bibitem{C97}
		{\sc Courcelle, B.}
		\newblock Handbook of graph grammars and computing by graph transformation.
		\newblock In {\em The Expression of Graph Properties and Graph Transformations
			in Monadic Second-order Logic}, G.~Rozenberg, Ed. World Scientific Publishing
		Co., Inc., River Edge, NJ, USA, 1997, pp.~313--400.
		
		\bibitem{Co11}
		{\sc Courcelle, B., and Engelfriet, J.}
		\newblock {\em Graph Structure and Monadic Second-Order Logic: A
			Language-Theoretic Approach}.
		\newblock Cambridge University Press, New York, NY, USA, 2012.
		
		\bibitem{Co00}
		{\sc Courcelle, B., Makowsky, J.~A., and Rotics, U.}
		\newblock Linear time solvable optimization problems on graphs of bounded
		clique-width.
		\newblock {\em Theory Comput. Sys. 33}, 2 (2000), 125--150.
		
		\bibitem{C93}
		{\sc Courcelle, B., and Mosbah, M.}
		\newblock Monadic second-order evaluations on tree-decomposable graphs.
		\newblock {\em Theor. Comput. Sci. 109}, 1-2 (1993), 49--82.
		
		\bibitem{Cowen86}
		{\sc Cowen, L., Cowen, R., and Woodall, D.}
		\newblock Defective colorings of graphs in surfaces: partitions into subgraphs
		of bounded valency.
		\newblock {\em J. Graph Theory 10\/} (1986), 187--195.
		
		\bibitem{Cowen97}
		{\sc Cowen, L., Goddard, W., and Jesurum, C.~E.}
		\newblock Defective coloring revisited.
		\newblock {\em J. Graph Theory 24}, 3 (1997), 205--219.
		
		\bibitem{CyganFKLMPPS15}
		{\sc Cygan, M., Fomin, F.~V., Kowalik, L., Lokshtanov, D., Marx, D., Pilipczuk,
			M., Pilipczuk, M., and Saurabh, S.}
		\newblock {\em Parameterized Algorithms}.
		\newblock Springer, 2015.
		
		\bibitem{Diestel10}
		{\sc Diestel, R.}
		\newblock {\em {Graph Theory}}, vol.~173.
		\newblock Springer-Verlag, 4th edition, 2010.
		
		\bibitem{Dorbec14}
		{\sc Dorbec, P., Montassier, M., and Ochem, P.}
		\newblock Vertex partitions of graphs into cographs and stars.
		\newblock {\em J. Graph Theory 75}, 1 (2014), 75--90.
		
		\bibitem{DF13}
		{\sc Downey, R.~G., and Fellows, M.~R.}
		\newblock {\em Fundamentals of Parameterized Complexity}.
		\newblock Texts in Computer Science. Springer, 2013.
		
		\bibitem{Eaton99}
		{\sc Eaton, N., and Hull, T.}
		\newblock Defective list colorings of planar graphs.
		\newblock {\em Bull. Inst. Combin. Appl 25\/} (1999), 79--87.
		
		\bibitem{fkr}
		{\sc Furma\'nczyk, H., Kubale, M., and Radziszowski, S.}
		\newblock On bipartization of cubic graphs by removal of an independent set.
		\newblock {\em Discrete Appl. Math. 209\/} (2016), 115--121.
		
		\bibitem{garcia16}
		{\sc Garc\'ia-V\'azquez, P.}
		\newblock On the bipartite vertex frustration of graphs.
		\newblock {\em Electronic Notes in Discrete Mathematics 54\/} (2016), 289 --
		294.
		
		\bibitem{gjs}
		{\sc Garey, M., Johnson, D., and Stockmeyer, L.}
		\newblock Some simplified \texttt{NP}-complete graph problems.
		\newblock {\em Theory Comput. Sci. 1}, 3 (1976), 237--267.
		
		\bibitem{Gimbel10}
		{\sc Gimbel, J., and Ne\v{s}et\v{r}il, J.}
		\newblock Partitions of graphs into cographs.
		\newblock {\em Discrete Math. 310}, 24 (2010), 3437 -- 3445.
		
		\bibitem{G00}
		{\sc Golumbic, M.~C., and Rotics, U.}
		\newblock On the clique-width of some perfect graph classes.
		\newblock {\em Int. J. Found. Comput. Sci. 11}, 03 (2000), 423--443.
		
		\bibitem{gfpp}
		{\sc Guillemot, S., Havet, F., Paul, C., and Perez, A.}
		\newblock On the (non-)existence of polynomial kernels for $p$-free edge
		modification problems.
		\newblock {\em Algorithmica 65}, 4 (2012), 900--926.
		
		\bibitem{Harary85}
		{\sc Harary, F., and Jones, K.}
		\newblock Conditional colorability ii: Bipartite variations.
		\newblock In {\em Proc. Sundance Cont. Combinatorics and related topics, Congr.
			Numer.\/} (1985), vol.~50, pp.~205--2018.
		
		\bibitem{ht74}
		{\sc Hopcroft, J., and Tarjan, R.}
		\newblock Efficient planarity testing.
		\newblock {\em J. ACM 21}, 4 (1974), 549--568.
		
		\bibitem{L12}
		{\sc Lampis, M.}
		\newblock Algorithmic meta-theorems for restrictions of treewidth.
		\newblock {\em Algorithmica 64}, 1 (2012), 19--37.
		
		\bibitem{Lima17}
		{\sc Lima, C.~V., Rautenbach, D., Souza, U.~S., and Szwarcfiter, J.~L.}
		\newblock Decycling with a matching.
		\newblock {\em Infor. Proc. Letters 124\/} (2017), 26 -- 29.
		
		\bibitem{lovasz66}
		{\sc Lov\'{a}sz, L.}
		\newblock On decomposition of graphs.
		\newblock {\em Studia Sci. Math. Hungar. 1\/} (1966), 237--238.
		
		\bibitem{mr}
		{\sc Mulzer, W., and Rote, G.}
		\newblock Minimum-weight triangulation is \texttt{NP}-hard.
		\newblock {\em J. ACM 55}, 2 (2008), 1--29.
		
		\bibitem{nss}
		{\sc Natanzon, A., Shamir, R., and Sharan, R.}
		\newblock Complexity classification of some edge modification problems.
		\newblock {\em Discrete Appl. Math. 113}, 1 (2001), 109--128.
		
		\bibitem{Fabio18}
		{\sc Protti, F., and Souza, U.~S.}
		\newblock {Decycling a graph by the removal of a matching: new algorithmic and
			structural aspects in some classes of graphs}.
		\newblock {\em {Discrete Mathematics \& Theoretical Computer Science} {vol. 20
				no. 2}\/} (2018).
		
		\bibitem{rs86}
		{\sc Robertson, N., and Seymour, P.}
		\newblock Graph minors. ii. algorithmic aspects of tree-width.
		\newblock {\em J. Algorith. 7}, 3 (1986), 309 -- 322.
		
		\bibitem{Schaefer}
		{\sc Schaefer, T.~J.}
		\newblock The complexity of satisfiability problems.
		\newblock In {\em STOC '78\/} (1978), pp.~216--226.
		
		\bibitem{T98}
		{\sc Thorup, M.}
		\newblock All structured programs have small tree width and good register
		allocation.
		\newblock {\em Inf. Comput. 142}, 2 (1998), 159--181.
		
		\bibitem{y2}
		{\sc Yannakakis, M.}
		\newblock Node-and edge-deletion \texttt{NP}-complete problems.
		\newblock In {\em STOC '78\/} (1978), pp.~253--264.
		
		\bibitem{y}
		{\sc Yannakakis, M.}
		\newblock Edge-deletion problems.
		\newblock {\em SIAM J. Comput. 10}, 2 (1981), 297--309.
		
		\bibitem{yr}
		{\sc Yarahmadi, Z., and Ashrafi, A.~R.}
		\newblock A fast algorithm for computing bipartite edge frustration number of
		$(3,6)$-fullerenes.
		\newblock {\em J. Theor. Comput. Chem. 13}, 02 (2014), 1450014--1450025.
		
	\end{thebibliography}
\end{document}